\documentclass[runningheads]{llncs}

\usepackage[T1]{fontenc}
\usepackage{graphicx}
\usepackage{color}
\usepackage{float}
\usepackage{amsmath}
\usepackage{tabularx}
\usepackage{booktabs}
\usepackage{amssymb}
\usepackage{stmaryrd}
\usepackage{cleveref}
\usepackage{multirow}
\usepackage{enumitem}
\usepackage{algorithm}
\usepackage{algorithmic}
\usepackage{tikz}
\usepackage{adjustbox}
\usepackage{mdframed}
\usepackage{subcaption}
\usepackage[table,xcdraw]{xcolor}

\usetikzlibrary{arrows.meta,positioning}
\tikzset{
	node/.style={draw, rounded corners, minimum width=2.5cm, minimum height=1cm, align=center},
	arrow/.style={-{Latex}, thick},
}
\usepackage{pgfplots}
\usepackage{pgfplotstable}
\pgfplotsset{compat=1.18}
\usepgfplotslibrary{groupplots,fillbetween}
\usepackage{tcolorbox}
\tcbuselibrary{skins, breakable}
\usepackage{caption}
\usepackage{wrapfig}

\newcommand{\R}{\mathbb{R}}
\newcommand{\Rplus}{\R_{\ge 0}}
\newcommand{\Rn}{\R^n}

\newcommand{\idn}{\mathbf{I}_n}

\newcommand{\normal}{\ensuremath{\mathcal{N}}}
\newcommand{\dirac}{\delta}
	
\newcommand{\id}{\text{Id}}

\newcommand{\density}[1]{f_{#1}}
\newcommand{\charfunc}{\mathbb{I}}

\newcommand{\symSkip}{\texttt{skip}}
\newcommand{\symSemi}{;~}
\newcommand{\symAsgn}{\texttt{:=}}
\newcommand{\symDist}{\ensuremath{\sim}}
\newcommand{\symIf}{\texttt{if}}
\newcommand{\symElse}{\texttt{else}}

\newcommand{\symObserve}{\texttt{observe}}
\newcommand{\blockStart}{\{}
\newcommand{\blockEnd}{\}}

\newcommand{\stmtSkip}{\symSkip}
\newcommand{\stmtSeq}[2]{{#1}\symSemi{#2}}
\newcommand{\stmtAsgn}[2]{{#1}~\symAsgn~{#2}}
\newcommand{\stmtDist}[2]{{#1}\symDist{#2}}

\newcommand{\stmtIf}[3]{\symIf~{#1}~\blockStart~{#2}~\blockEnd~\symElse~\blockStart~{#3}~\blockEnd}

\newcommand{\stmtObserve}[1]{\symObserve~{#1}}
\newcommand{\stmt}{P}

\newcommand{\gmstmt}{gm(\pi, \mu, \sigma)}
\newcommand{\gmg}{G_{\pi, \mu, \sigma}}

\newcommand{\bexp}{b}
\newcommand{\bexps}{\mathbb{B}}
\newcommand{\bexpsif}{\bexps_{\text{if}}}
\newcommand{\aexp}{e}
\newcommand{\aexps}{\mathbb{E}}

\newcommand{\params}{\Theta}

\newcommand{\entrytype}{\textit{entry}}
\newcommand{\detstatetype}{\textit{det}}
\newcommand{\rndstatetype}{\textit{rnd}}
\newcommand{\testtype}{\textit{test}}
\newcommand{\observetype}{\textit{observe}}
\newcommand{\exittype}{\textit{exit}}

\newcommand{\condfunc}{\textit{cond}}
\newcommand{\argfunc}{\textit{arg}}
\newcommand{\succfunc}{s_\path}

\newcommand{\node}{v}
\newcommand{\path}{\omega}
\newcommand{\pathset}[1]{\Omega^{#1}}

\newcommand{\p}{p}
\newcommand{\dist}{D}

\newcommand{\soga}[1]{\llbracket{#1}\rrbracket^S}
\newcommand{\exact}[1]{\llbracket #1 \rrbracket}
\newcommand{\mmop}{\mathcal{T}_\textrm{gm}}
\newcommand{\marg}[2]{\textrm{Marg}_{#1}(#2)}
\newcommand{\prob}[2]{P_{#1}(#2)}
\newcommand{\cond}[2]{#1_{\mid #2}}

\newcommand{\gm}{G}

\newcommand{\degas}[1]{\llbracket{#1}\rrbracket^{\Delta, \epsilon}}

\newcommand{\varset}{V}

\newcommand{\smoothEps}{\epsilon}

\newcommand{\smoothDelta}{\delta_\smoothEps}

\newcommand{\smoothOpBool}{\mathcal{B}_\smoothEps}

\renewcommand{\xi}{x_i}
\newcommand{\xj}{x_j}
\newcommand{\xk}{x_k}

\newcommand{\peps}{\p_{\smoothEps}}
\newcommand{\deps}{\dist_{\smoothEps}}
\newcommand{\meps}{\mu_{\smoothEps}}
\newcommand{\sigmaeps}{\Sigma_{\smoothEps}}
\newcommand{\vi}{\node_i}
\newcommand{\beps}{\bexp_{\smoothEps}}
\newcommand{\lpdist}{d_{LP}}

\newcommand{\degastool}{DeGAS\,}

\begin{document}
\title{DeGAS: Gradient-Based Optimization of Probabilistic Programs without Sampling}
%
\titlerunning{DeGAS}
%
\author{Francesca Randone\inst{1}\orcidID{0009-0002-3489-9600} \and
Romina Doz\inst{2}\orcidID{0009-0008-0210-6673} \and
Mirco Tribastone \inst{3}\orcidID{0000-0002-6018-5989} \and
Luca Bortolussi\inst{2}\orcidID{0000-0001-8874-4001} }

%
\institute{TU Wien, Wien, Austria \email{francesca.randone@tuwien.ac.at} \and
University of Trieste, Trieste, Italy \\
\and
IMT School for Advanced Studies Lucca, Lucca, Italy\\
}
\maketitle              
\begin{abstract}
We present DeGAS, a differentiable Gaussian approximate semantics for loopless probabilistic programs that enables sample-free, gradient-based optimization in models with both continuous and discrete components. DeGAS evaluates programs under a Gaussian-mixture semantics and replaces measure-zero predicates and discrete branches with a vanishing smoothing, yielding closed-form expressions for posterior and path probabilities. We prove differentiability of these quantities with respect to program parameters, enabling end-to-end optimization via standard automatic differentiation, without Monte Carlo estimators. On thirteen benchmark programs, DeGAS achieves accuracy and runtime competitive with variational inference and MCMC. Importantly, it reliably tackles optimization problems where sampling-based baselines fail to converge due to conditioning involving continuous variables.
%
\end{abstract}
\section{Introduction}

Probabilistic programming languages (PPLs) extend general-purpose languages with probabilistic primitives. In recent years they have received significant attention, motivated by their ability to model and analyze inherently stochastic systems such as randomized algorithms~\cite{karp1991introduction}, probabilistic cyber-physical systems~\cite{lee2016introduction}, and machine-learning models~\cite{borgstrom2011measure}.

Probabilistic programs (PPs) typically expose parameters that must be tuned to meet goals: maximize likelihood \cite{carpenter2017stan,bingham2019pyro}, satisfy safety constraints \cite{sato2019formal}, calibrate uncertainty \cite{chong2018guidelines}, or optimize bespoke utilities \cite{rainforth2016bayesian}. Current practice relies on sampling—either Markov chain Monte Carlo (MCMC) \cite{andrieu2003introduction} or variational inference (VI) \cite{blei2017variational}. While often effective, these methods can exhibit high-variance gradient estimates and brittle behavior on discontinuous objectives induced by branching and hard conditioning, frequently necessitating problem-specific reparameterization tricks \cite{DBLP:conf/nips/ZhangSSG12,DBLP:conf/nips/PakmanP13,DBLP:conf/nips/AfsharD15,DBLP:conf/icml/DinhBZM17,lee2018reparameterization}.

This paper asks the following: \emph{Can we obtain gradients for PP objectives without sampling?} We answer affirmatively by introducing DeGAS, a differentiable Gaussian-mixture approximate semantics that renders end-to-end evaluation differentiable, enabling optimization with off-the-shelf automatic differentiation (AD) and \emph{without} Monte Carlo estimators. The key idea is a principled, vanishing smoothing of discrete and measure-zero constructs (e.g., point masses and Boolean predicates), which yields differentiable densities and path probabilities throughout a program’s execution. This delivers deterministic objective values and gradients with respect to the program parameters.

DeGAS builds on the SOGA (second-order Gaussian approximation) semantics for a class of loop-free probabilistic programs~\cite{DBLP:journals/pacmpl/RandoneBIT24}. With SOGA, each program state is represented as a Gaussian mixture (GM) and program constructs update mixtures analytically (affine transforms, products, and conditioning via truncation) to approximate the posterior over program variables. DeGAS smooths the discontinuities that arise under branching and conditioning by injecting small Gaussian noise where needed and relaxing predicates accordingly, while preserving closure under mixture operations. We prove differentiability of posterior distributions and path probabilities with respect to parameters, and we show that as the smoothing parameter tends to zero, the differentiable semantics converges to the unsmoothed SOGA semantics under mild regularity conditions.

We implement DeGAS in PyTorch leveraging its AD capabilities to provide a sample-free gradient-based optimizer.  Numerical experiments show that, on thirteen benchmark programs, DeGAS achieves accuracy and runtime competitive with VI and MCMC. More important, we show DeGAS is able to solve instances of optimization problems for which sampling baselines fail to converge due to the presence of constructs such as conditional statements depending on continuous variables.


\paragraph{Structure of the paper.}
Section~\ref{sec:related} reviews related work. 
Section~\ref{sec:diff-semantics} presents the syntax and defines the differentiable semantics.
Section~\ref{sec:properties} states and proves our main properties (differentiability and convergence). 
Section~\ref{sec:experiments} reports the experimental results. Section~\ref{sec:conclusion} concludes.

\section{Related Work}
\label{sec:related}

To the best of our knowledge, this is the first work that computes gradients of probabilistic programs without sampling. For deterministic programs, the idea of smoothing to achieve sample-free optimization via gradient-free methods has been explored in~\cite{chaudhuri2010smooth,chaudhuri2011smoothing}, and successively adapted for automatic differentiation in~\cite{kreikemeyer2023smoothing}. However, as will be shown in Section~\ref{sec:properties}, the smoothing of deterministic programs does not trivially carry over to PPs.  

An attempt was made to extend the smoothing to PPs in Leios \cite{DBLP:conf/esop/LaurelM20}. Leios takes as input a discrete or hybrid discrete-continuous PP and approximates it with a fully continuous one using smoothing. The authors notice that the relaxation parameters must be tuned to avoid collapsing any probability event to zero. For example, a variable initially distributed as a Bernoulli after smoothing becomes distributed as a mixture of two components, with means in $0$ and $1$ and small standard deviations. This implies that the probability of the event $x==0$ passes from $0.5$ to $0$. To avoid this, Leios solves an optimization problem to replace all Booleans predicate depending on smoothed variables. Conceptually, we are close in spirit in that we tackle discontinuities by controlled smoothing. However, the goal is fundamentally different: similarly to the previously mentioned works, Leios aids sampling-based inference; instead, we enable sample-free optimization.

For PPs, most of the work has focused on coupling AD with sampling-based inference, designing gradient estimators for discontinuous programs—either to improve MCMC and VI globally \cite{DBLP:conf/nips/ZhangSSG12,DBLP:conf/nips/PakmanP13,DBLP:conf/nips/AfsharD15,DBLP:conf/icml/DinhBZM17,lee2018reparameterization} or only on the discontinuous parts \cite{DBLP:conf/aistats/ZhouGKRYW19,DBLP:conf/nips/KriekenTT21,DBLP:journals/pacmpl/LeeRY23}. 
These approaches still estimate gradients via sampling and are used to run inference faster or more stably. 
In contrast, we compute deterministic gradients by evaluating a differentiable program semantics; no Monte Carlo appears in the optimization loop. This lets us optimize likelihood and non-likelihood probabilistic objectives directly (such as reachability), not just expected values under a sampler.

Some methods start from an input program and compute a second program, whose expected return value is the derivative of the input program's expected return value~\cite{DBLP:conf/nips/AryaSSR22,DBLP:journals/pacmpl/LewHSM23}; 
the new program can then be sampled to obtain gradient estimations.
We differ in both goal (optimization rather than derivative-estimation for inference) and mechanism: for our smoothing, we leverage the analytical properties of the SOGA/GM semantics~\cite{DBLP:journals/pacmpl/RandoneBIT24}, that yield closed-form truncated moments and CDFs for every node, which we then backpropagate through analytically.

\section{Syntax and Semantics}\label{sec:diff-semantics}


SOGA was introduced in \cite{DBLP:journals/pacmpl/RandoneBIT24} as a semantics for a loop-free PPL that manipulates input variables distributed as GMs.  SOGA  makes the language closed under GMs, providing an analytical representation of the posterior that enables efficient computation of moments, cumulative distribution functions (CDFs), and probability density functions (PDFs).
This semantics is the second-order approximation of a family of approximate semantics based on Gaussian Mixtures, guaranteed to convergence to the true one, and demonstrated high accuracy across a variety of programs. Here, we add the possibility of specifying parameters to be used for optimization. 

\paragraph{Notation.} For a distribution $D$, let $\density{D}$ denote its pdf. 
A normal random variable with mean $\mu$ and covariance matrix $\Sigma$ is denoted by $\normal(\mu, \Sigma)$. A GM is a weighted sum of Gaussians, $\gm_{\pi, \mu, \Sigma} = \sum_{i=1}^c \pi_i \normal(\mu_i, \Sigma_i)$ where $c$ is the number of \textit{components} and $(\pi_i)_{\{i=1, \hdots, c\}}$ is the vector of $\textit{weights}$ satisfying $0 \le \pi_i \le 1$ and $\sum_{i=1}^c \pi_i = 1$. 
A Dirac delta centered on $x_0$ is denoted by $\dirac_{x_0}$.
Let $x[x_i \to e]$ denote the vector with $e$ as $i$-th component, and every other component equal to $x_j$. 
Let $\marg{y}{\dist}$ denote the marginal of distribution $\dist$ with respect to subvector $y$. 
We denote with $\idn$ the identity matrix of dimension $n$ and with $\charfunc_{A}$ the characteristic function of set $A$. 
We use $\otimes$ for the product of probability measures: if $\dist_1, \dist_2$ are distributions on $\R$, $\dist_1 \otimes \dist_2$ is a distribution on $\R^2$ whose marginal on the first coordinate is $\dist_1$ and the marginal on the second coordinate is $\dist_2$. 
For a predicate $\bexp$,  $\prob{\dist}{\bexp}$ denotes the probability of $\bexp$ under distribution $\dist$.
$\cond{\dist}{\bexp}$ denotes distribution $\dist$ conditioned to $\bexp$.


\subsection{Syntax}
\label{sec:syntax}

\paragraph*{Variables.} We consider programs defined over a vector of real variables $x = (x_1, \dots, x_n)$ and a set $\params$ of parameters, $\params = \{ \theta_i = d_i \}_{i = 1, \hdots, p}$ where $\theta_i$ is the name of the parameter and $d_i \in \R$ is an initial value. In addition, for each parameter an interval domain is specified.  
%
For instance, we can specify a standard deviation (std) as a parameter called $\sigma$, with initial value $1$, and interval domain $(0,+\infty)$.

\paragraph*{Expressions.} Expressions over program variables can either be products between two variables or linear expressions on program variables, where the coefficients of the variables or the zero-th order term can either be constants or parameters:
\begin{equation*}
	\aexps = \{ x_ix_j \mid i, j = 1 \dots n \} \cup \{ c_1x_1 + \hdots + c_nx_n + c_0 \mid c_i \in \mathbb{R} \cup \Theta \} \enspace.
\end{equation*}
\paragraph*{Boolean predicates.} Boolean predicates are \textit{true} or \textit{false}, or inequalities between a variable and a real number, which can either be a constant or a parameter:
\begin{equation*} 
	\bexps = \{true, false\} \cup \{ x_i \bowtie c \mid c \in \mathbb{R} \cup \Theta, \, \bowtie \in \{<, \le, ==, \ge, >\} \} 
\end{equation*}
Boolean predicates in the guards of conditional (if) statements are restricted to $\bexpsif$ with $\bowtie \in \{<, \le, \ge, > \}$.

The restrictions on expressions and Boolean predicates are inherited by SOGA and are motivated by the need to preserve closure of the Gaussian-mixture semantics under the supported operations. However, our grammar is expressive enough to encode general polynomial assignments and guards. 

\paragraph*{Language.} Our syntax uses the primitive $\stmtDist{x}{\gmstmt}$ for random assignments with GMs, where $\pi = (\pi_1, \hdots, \pi_c)$, $\mu = (\mu_1, \hdots, \mu_c)$, $\sigma = (\sigma_1, \hdots, \sigma_c)$ are vectors of weights, means and standard deviations, respectively, and $\pi_i , \mu_i, \sigma_i$
can either be real constants or real parameters, with the restriction that $\pi_i, \sigma_i \ge 0$ and $\sum_i \pi_i = 1$. The complete grammar of our programming language is:
\begin {align*}
\stmt ::= ~ & \stmtSkip \mid \stmtAsgn{x}{\aexp} \tag{deterministic assignments} \\
& \stmtDist{x}{\gmstmt} \tag{random assignment} \\
& \stmtIf{\bexp}{\stmt}{\stmt} \tag{test} \\
& \stmtObserve{\bexp} \tag{observe} \\
& \stmtSeq{\stmt}{\stmt}  \tag{sequential composition} 
\end {align*}
The $\stmtObserve{\bexp}$ construct represents a hard observe, which conditions the program’s distribution on the event specified by $\bexp$, effectively truncating the support of the resulting distribution to the states satisfying $\bexp$.

\subsection{Semantics}


We denote with $\exact{\cdot}$ the operator for the \emph{exact semantics} of the program,intended in the classic sense of Kozen's Semantics 2 \cite{kozen1979semantics}; instead, $\soga{\cdot}$ denotes the \emph{Second Order Gaussian Approximation (SOGA) semantics operator} as defined in \cite{DBLP:journals/pacmpl/RandoneBIT24} (both reported in Appedix \ref{app:soga} for completeness), which closes the language under GMs. The closure is realized by replacing every distribution that would not be a GM in the exact semantics (for instance, when conditioning due to a branch) with a GM. 
The substitution is performed by a dedicated operator called \emph{the moment-matching operator}, denoted by $\mmop$, which approximates a density function with a GM by matching its first two moments. 
In particular, when $\mmop$ acts on a mixture, it acts on every component; when it acts on a single-component mixture, it maps it to a Gaussian having mean and covariance matrix matching the original distribution, denoted by $\mu_\dist$, $\Sigma_\dist$. 
Formally:
\begin{equation}
	\mmop(\dist) = \begin{cases} \pi_1 \mmop(\dist_1) + \hdots + \pi_c \mmop(\dist_c) & \text{if } \dist=\sum_{i=1}^c \pi_i\dist_i \\ \normal(\mu_\dist, \Sigma_\dist) & \text{else.} \end{cases} \enspace .
\end{equation}
Notably, when acting on Gaussian mixtures, $\mmop$ leaves them unaltered~\cite{DBLP:journals/pacmpl/RandoneBIT24}.

\paragraph{Non-differentiability of SOGA semantics.}
By closing the semantics with respect to GMs, the general pdf of the distribution yielded by a program  dependent on parameters $\params$ can be shown to take the general form
$$ f_\params(x) = \sum_{i=1}^C \frac{\pi_i(\params)}{(2\pi)^{n/2}|\Sigma_i(\params)|^{1/2}}\exp\left(-\frac{1}{2}(x - \mu_i(\params))^T \Sigma_{i}^{-1}(\params)(x-\mu_i(\params))\right). $$
Assuming that $\pi(\params), \mu(\params)$ and $\Sigma(\params)$ are differentiable in $\params$, $f_\params$ can only be nondifferentiable in $\params$ if a covariance matrix is singular, i.e., $|\Sigma_i(\Theta)| = 0$.

\definecolor{codegray}{gray}{0.95}

\begin{figure}[t]
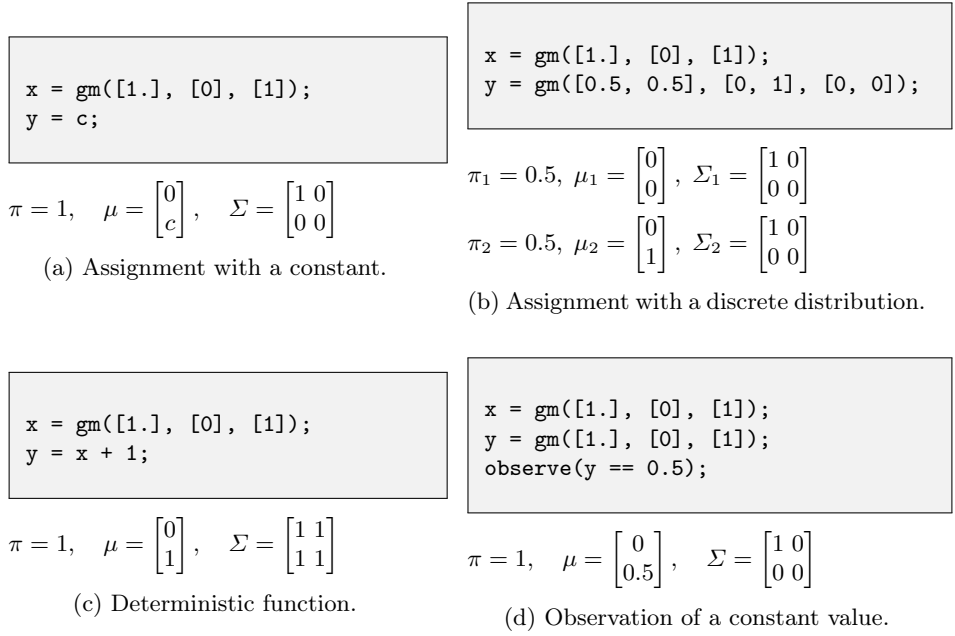

	\centering
	\begin{subfigure}[t]{0.45\textwidth}
		\centering
		\begin{minipage}{\linewidth}
			\vspace{0.5em}
			\fcolorbox{black}{codegray}{
				\parbox{\linewidth}{
					\texttt{
						\begin{tabbing}
							x = gm([1.], [0], [1]); \\
							y = c;
						\end{tabbing}
					}
			}}
			\vspace{0.5em}
			
			$\pi = 1, \quad
			\mu = 
			\begin{bmatrix}
				0 \\ c
			\end{bmatrix}, \quad
			\Sigma = 
			\begin{bmatrix}
				1 & 0 \\
				0 & 0
			\end{bmatrix}$
		\end{minipage}
		\caption{Assignment with a constant.}
		\label{fig:snippets-a}
	\end{subfigure}
	\hfill
	\begin{subfigure}[t]{0.5\textwidth}
		\centering
		\begin{minipage}{\linewidth}
			\vspace{0.5em}
			\fcolorbox{black}{codegray}{
				\parbox{\linewidth}{
					\texttt{
						\begin{tabbing}
							x = gm([1.], [0], [1]); \\
							y = gm([0.5, 0.5], [0, 1], [0, 0]);
						\end{tabbing}
					}
			}}
			\vspace{0.5em}
			
			$\pi_1 = 0.5, \;
			\mu_1 = 
			\begin{bmatrix}
				0 \\ 0
			\end{bmatrix}, \;
			\Sigma_1 = 
			\begin{bmatrix}
				1 & 0 \\ 0 & 0
			\end{bmatrix}$
			
			\vspace{0.3em}
			
			$\pi_2 = 0.5, \;
			\mu_2 = 
			\begin{bmatrix}
				0 \\ 1
			\end{bmatrix}, \;
			\Sigma_2 = 
			\begin{bmatrix}
				1 & 0 \\ 0 & 0
			\end{bmatrix}$
		\end{minipage}
		\caption{Assignment with a discrete distribution.}
		\label{fig:snippets-b}
	\end{subfigure}
	
	\vspace{1em}
	
	\begin{subfigure}[t]{0.45\textwidth}
		\centering
		\begin{minipage}{\linewidth}
			\vspace{0.5em}
			\fcolorbox{black}{codegray}{
				\parbox{\linewidth}{
					\texttt{
						\begin{tabbing}
							x = gm([1.], [0], [1]); \\
							y = x + 1;
						\end{tabbing}
					}
			}}
			\vspace{0.5em}
			
			$\pi = 1, \quad
			\mu = 
			\begin{bmatrix}
				0 \\ 1
			\end{bmatrix}, \quad
			\Sigma = 
			\begin{bmatrix}
				1 & 1 \\ 1 & 1
			\end{bmatrix}$
		\end{minipage}
		\caption{Deterministic function.}
		\label{fig:snippets-c}
	\end{subfigure}
	\hfill
	\begin{subfigure}[t]{0.5\textwidth}
		\centering
		\begin{minipage}{\linewidth}
			\vspace{0.5em}
			\fcolorbox{black}{codegray}{
				\parbox{\linewidth}{
					\texttt{
						\begin{tabbing}
							x = gm([1.], [0], [1]); \\
							y = gm([1.], [0], [1]); \\
							observe(y == 0.5);
						\end{tabbing}
					}
			}}
			\vspace{0.5em}
			
			$\pi = 1, \quad
			\mu = 
			\begin{bmatrix}
				0 \\ 0.5
			\end{bmatrix}, \quad
			\Sigma = 
			\begin{bmatrix}
				1 & 0 \\ 0 & 0
			\end{bmatrix}$
		\end{minipage}
		\caption{Observation of a constant value.}
		\label{fig:snippets-d}
	\end{subfigure}
	
	\caption{Nondifferentiability of SOGA semantics. Examples  yielding singular covariance matrices.}
	\label{fig:snippets}
\end{figure}

Figure \ref{fig:snippets} shows four minimal examples in which this may occur. 
In Figure \ref{fig:snippets-a}, $y$ is assigned a constant value $c$, setting its variance to $0$; in Figure \ref{fig:snippets-b}, $y$ is assigned with a discrete (Bernoulli) distribution, yielding again $0$ variance;
in Figure \ref{fig:snippets-c}, $y$ is assigned a deterministic function of $x$, so the distribution of $y$ is completely determined by that of $x$; reflected in the singularity of the covariance matrix.
In Figure \ref{fig:snippets-d}, $y$ is observed to take on a real value, having an effect analogous to a constant assignment.

In order to deal with such cases consistently, we smooth the semantics. Smoothing will depend on a hyper-parameter $\smoothEps > 0$ that we assume fixed for the rest of this section. Intuitively, $\smoothEps$ represents the amount of smoothing applied to the program and corresponds to the standard deviation of the Gaussian perturbation applied to smoothed variables. 

Next, we discuss how we smooth Boolean predicates and the whole semantics separately. Throughout the section, we let $\varset$ denote the set of program variables that have been smoothed. 




\paragraph{Smoothing operator for Boolean predicates.}

We define the following smoothing operator for Boolean predicates, that takes as input a predicate $\bexp$ and a set of smoothed variables $\varset$ and returns a new predicate.

\begin{equation} \label{eq:smoothOpBool}
	\smoothOpBool(x \bowtie c, \varset) =
	\begin{cases}
		x > c + \smoothDelta & \text{if } \bowtie \text{ is } >, x \in \varset \\
		x > c - \smoothDelta & \text{if } \bowtie \text{ is } \geq, x \in \varset \\
		x < c + \smoothDelta & \text{if } \bowtie \text{ is } \leq, x \in \varset \\
		x < c - \smoothDelta & \text{if } \bowtie \text{ is } <, x \in \varset \\
		(x > c - \smoothDelta) \wedge (x < c + \smoothDelta)
		& \text{if } \bowtie \text{ is } == t, x \in \varset \\
		x \bowtie c & \text{else}
	\end{cases}
\end{equation}
For now we allow $\smoothDelta$ to vary with $\smoothEps$ under the only assumption that $\lim_{\smoothEps \to 0} \smoothDelta = 0$. In the next section we will introduce an additional restriction to guarantee convergence to the SOGA semantics. 

\paragraph{DeGAS.}
As in \cite{DBLP:journals/pacmpl/RandoneBIT24}, we introduce the semantics in terms of the program's control flow graph (cfg). 
%
Each node in the graph represents a program instruction. 
We consider six node types: $\entrytype$, $\detstatetype$, $\rndstatetype$, $\testtype$, $\observetype$ and $\exittype$ and use the notation $\node \colon \textit{type}$ to denote that $\node$ is of type $\textit{type}$. 
Nodes of type $\detstatetype$ are labelled by $\stmtSkip$ or by a deterministic assignment $\stmtAsgn{\xi}{\aexp}$. Nodes of type $\rndstatetype$ are labelled by a random assignment $\stmtDist{\xi}{\gmstmt}$. Nodes of type $\testtype$ are labelled by a Boolean predicate in $\bexpsif$, while $\observetype$ node are labelled by a Boolean predicate in $\bexps$. A function $\argfunc$ is defined on all labelled nodes $\node$, returning the label of $\node$. In addition, a function $\condfunc$ is defined on the set of $\detstatetype$ and $\rndstatetype$ nodes, returning either \textit{true, false} or \textit{none} and such that $\condfunc(\node) = \textit{none}$ if and only if the parent of $\node$ is not a $\testtype$ node. Arithmetic expressions and Boolean predicates are interpreted in the standard way. An expression $\aexp \in \aexps$ denotes a mapping $e \colon \Rn \to \R$. A Boolean predicate $\bexp \in \bexps$ denotes the set of vectors $x \in \Rn$ that satisfy the predicate.

For a program $\stmt$ we let $\pathset{\stmt}$ denote the set of paths in its cfg. 
For each path $\path \in \pathset{\stmt}$ the function $\succfunc(\node)$ associates $\node \in \path$ with its succcessor in $\path$.
The semantics of path $\path$, denoted by $\degas{\path}$ is defined as a function taking no input and returning a pair $(\p, \dist)$  where $\p \in \Rplus$ and $\dist$ is a GM on $\Rn$. The semantics of a path $\path = \node_0..\node_n$ is defined as the composition of the semantics of each node in the path: $ \degas{\path} = \degas{v_n}_\path \circ \hdots \circ \degas{v_0}_\path.$

As noted earlier, when computing the semantics of a node, in addition to $\p$ and $\dist$, we need to keep track of the smoothed variables, so we augment the semantics with a set $\varset$, containing the variables that would be discrete in the exact semantics.
The semantics of a node $\node$ is a function taking as input a triple $(\p, \dist, \varset)$, and returning a triple  $(\p', \dist', \varset')$ of the same type. 
The only exceptions are the $\entrytype$ node which takes no input and the $\exittype$ node which just outputs the pair $(\p, \dist)$ The semantics of each node type is defined separately.



\begin{itemize}
	\item If $\node \colon \entrytype$, then  $\degas{\node}_\path = (1, \normal(0, \mathbb{I}_n), \emptyset).$
	\item If $\node \colon \detstatetype$, let $\dist'$ denote the distribution of $x[\xi \to \aexp]$ and $\dist^\smoothEps$ denote the distribution of $x[\xi \to \aexp + z]$, where $z$ is a fresh program variable with distribution $\normal(0, \smoothEps)$. Let $\varset'$ be $\varset \cup \{\xi\}$ if all variables in $\aexp$ are in $\varset$ and $\varset' = \varset$ else. Then:
	$$ \degas{\node}_{\path}(\p, \dist, \varset) = \begin{cases} 
		(\p, \dist, \varset) & \argfunc(\node) = \stmtSkip \\ 
		(\p, \dist^\smoothEps, \varset \cup \{ \xi \} ) & \argfunc(\node) = \stmtAsgn{\xi}{c} \\
		(\p, \dist^\smoothEps, \varset')  & \argfunc(\node) = \stmtAsgn{\xi}{\aexp}, \aexp \text{ linear}, \xi \text{ not in } e \\
		(\p, \mmop(\dist'), \varset') & \text{else}. \end{cases}$$ 
	\item If $\node \colon \rndstatetype$ and $\argfunc(\node) = \stmtDist{\xi}{\gmstmt}$, let $\sigma'$ be the vector with $\sigma'_j = \sigma_j$ if $\sigma_j \neq 0$ and $\sigma'_j = \smoothEps$ else. Then:
	$$\degas{\node}_\path(\p, \dist, \varset) = \begin{cases} (\p, \marg{x \setminus \xi}{D} \otimes \gm_{\pi, \mu, \sigma'}, \varset \cup \{\xi\}) & \sigma_j = 0 \text{ for some } j \\
		(\p, \marg{x \setminus \xi}{D} \otimes \gm_{\pi, \mu, \sigma}, \varset) & \text{else}. \end{cases}$$
	\item If $ \node \colon \testtype$, letting $\argfunc(\node) = \bexp$ and $\bexp' = \smoothOpBool(\bexp, \varset)$:
	$$ \degas{\node}_\path(\p, \dist, \varset) = \begin{cases} 
		(\p \cdot \prob{\dist}{\bexp'}, \mmop(\cond{\dist}{\bexp'}), \varset) & \condfunc(\succfunc(\node)) = \textit{true}\\
		(\p \cdot \prob{\dist}{\neg \bexp'}, \mmop(\cond{\dist}{\neg \bexp'}), \varset) & \condfunc(\succfunc(\node)) = \textit{false}. \end{cases} $$
	\item If $ \node \colon \observetype$, letting $\argfunc(\node) = \bexp$, $I = \int_{\R^{n-1}} \density{\dist}(x[\xi \to c]) d(x\setminus \xi)$ and $\bexp' = \smoothOpBool(\bexp, \varset)$:
	$$ \degas{\node}_\path(\p, \dist, \varset) = \begin{cases}
		(\p \cdot I, \marg{x\setminus\xi}{\cond{\dist}{\bexp}} \otimes \normal(c, \smoothEps), \varset \cup \{ \xi \}) & \\
		& \hspace{-1cm} \bexp \text{ is } \xi == c, \xi \not \in \varset \\
		(\p \cdot \prob{\dist}{\bexp'}, \mmop(\cond{\dist}{\bexp'}), \varset) & \hspace{-1cm} \text{else} 
	\end{cases} $$
	\item If $\node \colon \exittype$, $\degas{\node}_\path(\p, \dist, \varset) = (\p, \dist)$.
\end{itemize}

A program $\stmt$ is called \textit{valid} if for at least a path $\path \in \pathset{\stmt}$ $\exact{\path} = (\p, \dist)$ with $\p > 0$. The semantics of a valid program $\stmt$ is then defined as a mixture of the semantics of its path:
\begin{equation*}%
	\degas{\stmt} = \sum_{\substack{\path \in \pathset{\stmt} \\  \degas{\path} = (\p, \dist)}} \frac{\p \cdot \dist}{\sum_{\substack{\path' \in \pathset{\stmt} \\ \degas{\path'} = (\p', \dist')}} \p'}.
\end{equation*}

Let us go back to Figure \ref{fig:snippets}, to briefly discuss how $\degas{\cdot}$ avoids the discontinuities.
For the program in Figure \ref{fig:snippets-a}, $\degas{\cdot}$ effectively assigns $y$ with $c + z$, where $z$ is a freshly defined variable, Gaussianly distributed with mean $0$ and std $\smoothEps$. This is equivalent to assigning $y$ with $\normal(c, \smoothEps)$, and yields a non-degenerate covariance matrix $\Sigma = \begin{bmatrix} 1 & 0 \\ 0 & \smoothEps^2 \end{bmatrix}$.
For the program in Figure \ref{fig:snippets-b}, $\degas{\cdot}$ replaces components with zero std with non-degenerate components with std $\smoothEps$. Therefore the assignment effectivaly performed would be $y = gm([0.5, 0.5], [0, 1], [\smoothEps, \smoothEps])$, yielding $\Sigma_1 = \Sigma_2 = \begin{bmatrix} 1 & 0 \\ 0 & \smoothEps^2 \end{bmatrix}$.
For the program in Figure \ref{fig:snippets-c}, $\degas{\cdot}$ acts similarly as in the first case, $y$ is a assigned with $x + 1 + z$, with $z$ fresh Gaussian variable with mean $0$ and std $\smoothEps$. The covariance matrix then becomes $\Sigma = \begin{bmatrix} 1 & 1 \\ 1 & 1 + \smoothEps^2 \end{bmatrix}$.
Finally, for the program in Figure \ref{fig:snippets-d}, the variance of $y$ after the observation is corrected and put equal to $\smoothEps$, yielding $\Sigma = \begin{bmatrix} 1 & 0 \\ 0 & \smoothEps^2 \end{bmatrix}$.

%
%
%
To avoid unwanted changes in the behaviour of the probability mass due to the smoothing, $\degas{\cdot}$ uses the operator $\smoothOpBool$ to interpret differently Boolean predicates when they involve smoothed variables, as reflected in the definition of the semantics of $\testtype$ and $\observetype$ nodes. 
This has the effect of avoiding conditioning to zero probability events, since $\degas{\cdot}$ always acts on non-degenerate distributions and $\smoothOpBool(\bexp, \varset)$ always returns open sets. 

In the next section, we will prove that the semantics computed by $\degas{\cdot}$ is indeed differentiable in the program parameters and that by suitably choosing the value of $\smoothDelta$ when applying $\smoothOpBool$ $\degas{\cdot}$ will converge to to $\soga{\cdot}$.

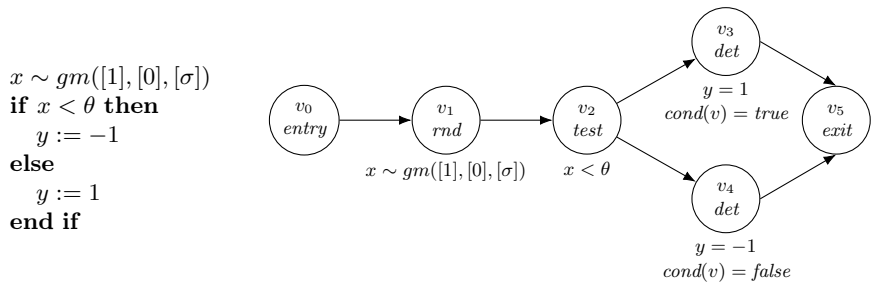
\begin{figure}[htbp]
\centering
\begin{minipage}{0.3\linewidth}
\begin{algorithmic}
    \STATE $x \sim gm([1], [0], [\sigma])$
    \IF{$x < \theta$}
        \STATE $y := -1$
    \ELSE
        \STATE $y := 1$
    \ENDIF
\end{algorithmic}
\end{minipage}
\begin{minipage}{0.6\linewidth}
\centering
\begin{tikzpicture}[node distance=1.8cm, scale=0.8, every node/.style={transform shape}]
    \tikzset{
        circnode/.style={
            draw, circle, minimum width=1cm, minimum height=1cm, inner sep=0pt, 
            text width=1cm, align=center, font=\small
        },
        arrow/.style={-{Latex}}
    }

    \node[circnode] (entry) {$\node_0$ \\ $\entrytype$};

    \node[circnode, right=1.2cm of entry, label=below:{\shortstack{$x \sim gm([1], [0], [\sigma])$}}] 
        (rnd) {$\node_1$ \\ $\rndstatetype$};

    \node[circnode, right=1.2cm of rnd, label=below:{$x < \theta$}] 
        (test) {$\node_2$ \\ $\testtype$};

    \node[circnode, below right=0.5cm and 1.5cm of test, 
          label=below:{\shortstack{$y=-1$ \\ $\condfunc(\node)=\textit{false}$}}] 
          (det2) {$\node_4$ \\ $\detstatetype$};

    \node[circnode, above right=0.5cm and 1.5cm of test, 
          label=below:{\shortstack{$y=1$ \\ $\condfunc(\node)=\textit{true}$}}] 
          (det1) {$\node_3$ \\ $\detstatetype$};

    \node[circnode, right=3cm of test] 
          (exit) {$\node_5$ \\ $\exittype$};

    \draw[arrow] (entry) -- (rnd);
    \draw[arrow] (rnd) -- (test);
    \draw[arrow] (test) -- (det1);
    \draw[arrow] (test) -- (det2);
    \draw[arrow] (det1.east) -- (exit.north);
    \draw[arrow] (det2.east) -- (exit.south);
\end{tikzpicture}
\end{minipage}
\caption{Example program $P$ (left) and its cfg (right).}\label{fig:cfg}
\end{figure}

\paragraph{Example 1.}
Consider the program $P$ and its cfg in Fig. \ref{fig:cfg}, where $\theta$ and $\sigma$ are parameters with $\theta \in (-\infty, +\infty)$ and $\sigma \in (0, +\infty)$. 
The exact semantics is:
\[
    \exact{\stmt} = \Phi(\theta/\sigma) \cdot (\cond{\normal(0, \sigma)}{x < \theta} \otimes \dirac_{-1}) 
    + (1-\Phi(\theta/\sigma)) \cdot (\cond{\normal(0, \sigma)}{x \ge \theta} \otimes \dirac_{1})
\]
In particular, the exact marginal over $y$ is a discrete distribution putting
$\Phi(\theta/\sigma)$ probability mass on $-1$ and $1-\Phi(\theta/\sigma)$ on $1$.
Therefore, the exact marginal over $y$ is nondifferentiable (in fact, discontinuous) in the parameters.

Using the SOGA semantics, closed with respect to GMs, yields:
\[
  \soga{\stmt} = \Phi(\theta/\sigma) \cdot \mmop(\cond{\normal(0, \sigma)}{x < \theta} \otimes \dirac_{-1})
    + (1-\Phi(\theta/\sigma)) \cdot \mmop(\cond{\normal(0, \sigma)}{x \ge \theta} \otimes \dirac_{1})
\]
However, the marginal over $y$ remains unchanged, and therefore still nondifferentiable.
When computing the differentiable semantics, not only is the moment-matching operator applied to close the semantics with respect to GMs, degeneracies are avoided by smoothing the assignments of $y$.
This yields:
\begin{align*}
  \degas{\stmt} &= \Phi(\theta/\sigma) \cdot \mmop(\cond{\normal(0, \sigma)}{x < \theta} \otimes \normal(-1, \smoothEps)) + \\
  & \hspace{4cm} + (1-\Phi(\theta/\sigma)) \cdot \mmop(\cond{\normal(0, \sigma)}{x \ge \theta} \otimes \normal(1, \smoothEps))
\end{align*}
For any $\sigma \in (0,+\infty)$, the differentiable marginal over $y$ is given by
\[
  \Phi(\theta/\sigma) \cdot \normal(-1, \smoothEps) + (1-\Phi(\theta/\sigma)) \cdot \normal(1, \smoothEps).
\]

\begin{figure}
    \centering
    \includegraphics[width=1.0\linewidth]{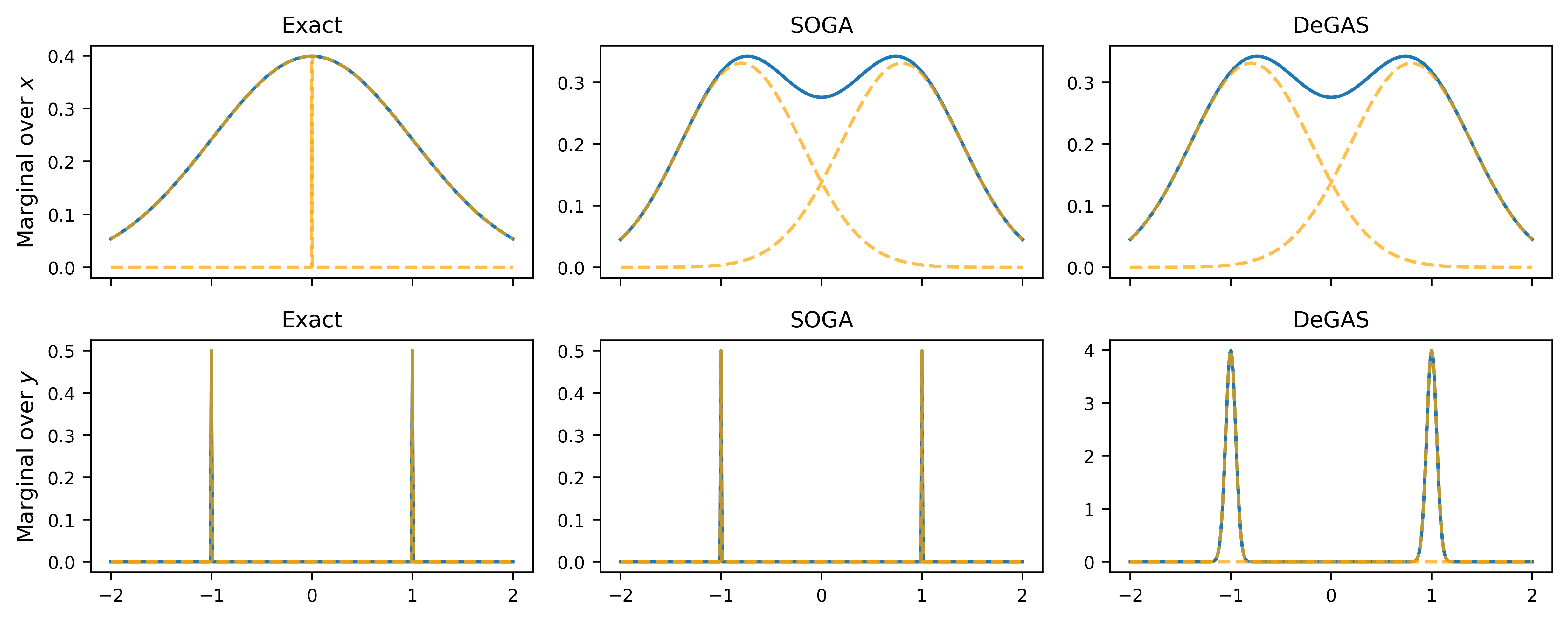}
    \caption{Marginal posterior densities of program $P$ under different semantics. The orange lines highlight the different components of the mixtures.}
    \label{fig:densities}
\end{figure}
The marginal densities of the three semantics, computed for $\sigma = 1$ and $\theta=0$ are represented in Fig. \ref{fig:densities}, where the densities of the components of the mixtures are highlighted in orange and for DeGAS we fixed $\smoothEps=0.05$.

\section{Properties of $\degas{\cdot}$}
\label{sec:properties}

In this section we prove two properties of the DeGAS semantics operator $\degas{\cdot}$, namely differentiability of $\degas{\stmt}$ with respect to the parameters and convergence of $\degas{\stmt}$ to $\soga{\stmt}$ as $\smoothEps \to 0$.

\paragraph*{Differentiability.} First, we focus on the differentiability of the DeGAS semantics with respect to $\params$. 
For this reason, in this subsection we make the dependence of the programs on the parameters explicit, using the notation $\stmt(\params)$.
For instance, program $\stmt$ from Example 1 would be denoted as $\stmt(\theta, \sigma)$.
We assume that the parameters come with assigned initial values and interval domains and that they always take values in the interval domains.
For instance, for $\stmt(\theta, \sigma)$ we could specify the initial values $d_\theta = 0$ and $d_\sigma = 1$ and the interval domains $\theta \in (-\infty, +\infty)$, $\sigma \in (0, +\infty)$.  
We say that a distribution $\dist(\params)$ is differentiable in $\params$ if its pdf is. For example, if $\dist(\params)$ is a GM, $\dist(\params)$ is differentiable in $\params$ if, for $i=1, \hdots, c$ (i) $\pi_i(\params)$, $\mu_i(\params)$ and $\Sigma_i(\params)$ are differentiable in $\params$; (ii) $\Sigma_i(\params)$ is non-singular for every value of $\params$ that satisfies the constraints. 

\begin{theorem} \label{thm:differentiability}
	For any valid program $\stmt(\params)$, $\degas{\stmt(\params)}$ returns $(\p'(\params), \dist'(\params), \varset')$ such that $\p'(\params)$ and $\dist'(\params)$ are differentiable in $\params$.
\end{theorem}
This is the core theoretical result of the paper (proved in the Appendix), because it allows for a safe application of gradient-based optimization methods to PPs.

\paragraph*{Convergence.} DeGAS adds an additional layer of approximation, through the Gaussian perturbation, to the already approximated semantics of SOGA $\soga{\cdot}$.  We prove that this additional approximation tends to disappear when the perturbation becomes very small, i.e. in the limit for $\smoothEps \to 0$. Coupled with SOGA's convergence result to the true semantics $\exact{\cdot}$, this offers a principled way of obtaining sample-free gradients for a class of PPs. 

DeGAS' convergence proof (provided in the Appendix) is surprisingly subtle and deserves an intuitive explanation. When smoothing the predicates, the dependence of $\smoothDelta$ on $\smoothEps$ can jeopardize the convergence. The following example motivates the introduction of an additional requirement, called \emph{consistency}, that is needed to guarantee convergence.

\begin{example}
	We consider two example programs:
	\begin{align*} 
		\stmt_1 \colon & \stmtAsgn{x}{0}; \stmtDist{y}{gm([1.], [0.], [1.])}; \stmtObserve{x \ge 0}, \\
	 	\stmt_2 \colon & \stmtAsgn{x}{0}; \stmtDist{y}{gm([1.], [0.], [1.])}; \stmtObserve{x > 0}.
	 \end{align*}
	 In SOGA we have $ \soga{\stmt_1} = (1, \delta_0 \times \normal(0,1))$ and $\soga{\stmt_2} = (0, \delta_0 \times \normal(0,1))$.
	 In DeGAS, before the observe statement for both programs we have the output triple $(1, \deps, \{x\})$ where $\deps = \normal(0, \smoothEps) \times \normal(0, 1)$.
	 However, the predicates of the observe node are smoothed differently:
	 $$ \smoothOpBool(x \ge 0, \{x\}) = x > -\smoothDelta, \qquad \smoothOpBool(x > 0, \{x\}) = x > \smoothDelta.$$
	 Consider $\stmt_1$. 
	 To compute the differentiable semantics of the observe node we need to compute $\prob{\deps}{x > -\smoothDelta}$ and the moments of $\cond{(\deps)}{x > \smoothDelta}$, and we want them to tend to $1$ and to the moments of $\delta_0 \times \normal(0,1)$, respectively. 
	 Here, we focus on the probabilities, since the convergence of the moments are proved using similar arguments.
	 Letting $\Phi$ denote the cdf of the standard Gaussian, then $ \prob{\deps}{x > - \smoothDelta} = 1 - \Phi\left(-\frac{\smoothDelta}{\smoothEps}\right).$
	 For this quantity to tend to 1 we need $\frac{\smoothDelta}{\smoothEps} \to \infty$.
	 Intuitively, this limit forces $\smoothDelta$ to decrease slower than the variance of $x$ (in this case $\smoothEps$), so that as $\smoothEps$ tends to $0$ more and more of the probability mass of $\deps$ is inside $\{ x > -\smoothDelta \}$. 
	 Moreover, it is clear that if $\smoothDelta$ is a linear function of $\smoothEps$ or if it decreases faster than it, the limit cannot hold.
	  
	 Analogously consider $\stmt_2$ for which $\prob{\deps}{x > \smoothDelta} = 1 - \Phi(\frac{\smoothDelta}{\smoothEps})$. For it to tend to $0$ we need again $\frac{\smoothDelta}{\smoothEps} \to \infty$. This time the slower convergence of $\smoothDelta$ to $0$ allows a larger part of the probability mass of $\deps$ to remain outside $\{ x > \smoothDelta \}$.
\end{example}
	 
The previous example leads us to the following definition and theorem (we report the full proof in the Appendix).

\begin{definition}
	The operator $\smoothOpBool$ is applied \emph{consistently} if for every $\testtype$ and $\observetype$ node with $\argfunc(\node) = \xi \bowtie c$ when $\smoothOpBool$ is applied to compute $\degas{\node}(\p, \dist, \varset)$ and $\xi \in \varset$, $\delta$ is chosen such that:
	$$ \lim_{\smoothEps \to 0} \smoothDelta(\epsilon) = 0 \qquad \text{ and } \qquad \lim_{\smoothEps \to 0} \frac{\smoothDelta(\smoothEps)}{\sqrt{\Sigma_{\xi, \xi}(\epsilon)}} = +\infty $$
	for every $\Sigma(\smoothEps)$ covariance matrix of a component of $\dist$. 
	The dependence of $\Sigma$ from $\smoothEps$ is guaranteed by the fact that $\xi \in \varset$.  
\end{definition}

\begin{theorem}
	Suppose $\stmt$ is a valid program and $\smoothOpBool$ is applied consinstently when computing $\degas{\stmt}$. Then $\lim_{\smoothEps \to 0} \degas{\stmt} = \soga{\stmt}$ (where convergence for the distributions is intended in the sense of \emph{convergence in distribution}).
\end{theorem}

 We remark that, while in theory it is possible to let $\epsilon$ tend to $0$ to make the additional error vanish, in practice we want to use a value of $\smoothEps$ large enough to enable efficient gradient-based optimization, as will be discussed next. 
\section{Experimental Evaluation}
\label{sec:experiments}

\paragraph*{Setup.} We evaluate DeGAS along two axes: (i) we optimize parameters of benchmark probabilistic programs taken from the literature comparing  with Variational Inference (VI) and Markov Chain Monte Carlo (MCMC); (ii) 
we evaluate DeGAS on programs with discontinuities related to continuous variables, which, as we will show, make standard optimization methods inapplicable. 
In this case we consider classic models of cyber-physical systems (CPS), whose parameters are synthesized to optimize trajectory likelihoods or reachability properties. 

We implemented DeGAS in Pytorch, leveraging automatic differentiation for the gradient computation (cf. Appendix~\ref{sec:degas}) and employing Adam as optimizer. 
For all the experiments we set $\smoothDelta = \sqrt{\smoothEps}$ and report the results for $\smoothEps = 0.001$. 
The choice of $\smoothDelta$ is coherent for the consistency assumptions. 
Furthermore, we did not observe significant accuracy changes with respect to the SOGA approximation when varying $\smoothEps$, as long as it stays below $0.1$, as reported in Appendix \ref{app:conv}. 
This is further confirmed by the results reported in Appendix \ref{app:SOGAcomp}, pointing to the fact that the main source of error is that induced by SOGA, while $\smoothEps$ has a negligible empirical effect for the considered values. 
For a dedicated analysis of the SOGA error for inference tasks, including discussion on increasing the Gaussian approximation order, we refer the reader to \cite{DBLP:journals/pacmpl/RandoneBIT24}.
We run all the experiments on a machine equipped with an Apple M2 processor with 8 cores running at 3.49\,GHz and 8\,GB  RAM. 
For all the experiments we set a time-out threshold of $600$s.


\paragraph*{Comparison with VI and MCMC.}
We start by comparing DeGAS with Variational Inference (VI) and Markov Chain Monte Carlo (MCMC) as implemented in Pyro \cite{bingham2019pyro} on a set of $13$ case studies taken from the PPL literature (\cite{carpenter2017stan,gehr2016psi,huang2021aqua}). 
As a standard benchmark optimization task, we also consider a Proportional Integral
Derivative (PID) controller, described in Appendix \ref{app:pid}. 

For each case study, some of the significant variables are set as optimizable parameters while others as observables. 
For the observables a dataset $D$ of size $N = 1000$ is generated for the true value of the parameters. 
For the PID case study, data correspond to idealized representations of stable dynamics, maintaining the system at its equilibrium point. 
For each program, the posterior density computed by DeGAS is used to compute the negative log-likelihood $l(\params) = -\log f_{\params}(x)$ of the dataset $D$; then, we optimize to find $\params^* = \text{argmin}_\params l(\params).$

For both VI and DeGAS, we report the results for the value of the learning rate chosen among $\{0.001, 0.005, 0.01, 0.05, 0.1, 0.2\}$ that achieves the minimum error. 
 The optimization is stopped upon convergence within a tolerance of $10^{-8}$ with a patience of $30$ iterations, with a maximum of $500$ steps for DeGAS (except $1000$ for the PID model) and $1000$ for VI. 
 For VI, we use the \texttt{Trace\_ELBO} loss\footnote{\url{https://docs.pyro.ai/en/dev/inference\_algos.html}}.
 For MCMC, we run the NUTS kernel with $4$ chains, initially running inference with $500$ samples and $50$ warm-up steps. If $\hat{R}$ exceeds $1.05$, the procedure is repeated, increasing the number of samples by $500$ and the warm-up steps by $50$ when the total number of steps is between $2000$ and $5000$, and by $100$ thereafter, up to a maximum of $6000$ total steps. All hyperparameter values are reported in Appendix \ref{app:hyper}. 

Table \ref{tab:vi_mcmc_degas} collects the results, where `time' refers to the average runtimes (in s) out of 10 executions and `error' refers to the average relative difference between the optimized value and the real value across all the optimizable parameters of the program. As the table shows, our approach generally achieves performance comparable to the baseline methods, and in several benchmarks it outperforms them in terms of accuracy. Two models particularly challenging for DeGAS are Grass and Burglary. Both models contain a large number of nested if statements, that give rise to a high number of components in the GM approximation. This problem is also known for SOGA and can be mitigated using a pruning strategy~\cite{DBLP:journals/pacmpl/RandoneBIT24}, whose tuning is outside the scope of this paper. 

\begin{table}[t]
\centering
\begin{tabular}{l|c|c|cc|cc|cc}
\toprule
\textbf{Model} & \textbf{\# Par.} & \textbf{\# Var.} &
\multicolumn{2}{c|}{\textbf{VI}} & 
\multicolumn{2}{c|}{\textbf{MCMC}} & 
\multicolumn{2}{c}{\textbf{DeGAS}} \\
\cmidrule(lr){4-5} \cmidrule(lr){6-7} \cmidrule(lr){8-9}
 & & & \textbf{Time} & \textbf{Error}
 & \textbf{Time} & \textbf{Error} 
 & \textbf{Time} & \textbf{Error} \\
\midrule
\rowcolor{gray!20}
Bernoulli & $1$ & $2$ & $0.204$ & $0.024$ & $10.968$  & $0.023$& $1.986$ & $0.001$ \\
Burglary & $2$ & $7$  & $3.244$ & $0.913$ & $13.014$ & $0.141$  & $14.347$ & $1.686$ \\
\rowcolor{gray!20}
ClickGraph & $1$ & $6$ & $1.635$ & $0.376$ & $62.546$ & $0.096$ & $21.600$ & $0.086$\\
\rowcolor{gray!20}
ClinicalTrial  & $3$ & $6$ & $1.055$ & $1.980$ & $120.492$ & $2.536$ &  $17.070$ & $0.657$\\
Grass  & $4$ & $6$ & $3.077$ & $0.306$ & $58.151$ & $0.015$  & $317.582$ & $0.036$\\
\rowcolor{gray!20}
MurderMistery  & $1$ & $1$ & $0.792$ & $8.015$ & $14.036$ & $2.303$  & $2.508$ & $0.203$\\
\rowcolor{gray!20}
SurveyUnbiased  & $2$ & $2$ & $0.794$ & $0.013$ & $17.030$ & $0.013$  & $11.570$ & $0.008$\\
\rowcolor{gray!20}
TrueSkills & $3$ & $6$ & $2.056$ & $0.003$ & t.o. & t.o. & $1.664$ & $0.003$\\
\rowcolor{gray!20}
TwoCoins & $2$ & $1$ & $0.095$ & $0.864$ & $49.120$ & $0.873$  & $3.410$ & $0.688$\\
\rowcolor{gray!20}
AlterMu & $3$ & $5$ & $1.007$ & $1.043$ & t.o. & t.o. & $0.640$ & $0.312$\\
\rowcolor{gray!20}
AlterMu2 & $2$ & $3$ & $0.718$ & $0.556$ & $384.730$ & $0.139$  & $0.618$ & $0.096$\\
NormalMixtures & $3$ & $3$ & $
1.061$ & $0.909$ & $345.409$ & $0.053$ & $2.941$ & $0.386$\\
PID & $2$ & $55$ & $10.778$ & -- & nc & nc & $213$ & -- \\
\bottomrule
\end{tabular}
\caption{Comparison of inference methods across models.
Models highlighted in gray indicate cases where DeGAS outperformed both other methods in either accuracy or execution time. For the last model, PID, we are not able to quantify the error, as the true parameters are unknown, but  we show the result of the optimization in Figure \ref{fig:PID_opt_reach}. MCMC fails to converge for this case study ($\hat{R} \gg 1$.)}
\label{tab:vi_mcmc_degas}
\end{table}

Many benchmark programs include conditional statements with Boolean guards over discrete variables. When these conditions are enumerable, both VI
\begin{wrapfigure}{r}{0.35\textwidth}
\vspace{-1em} 
\begin{minipage}{\linewidth}
\begin{algorithmic}
    \STATE $v \sim gm([1], [\mu_1], [5])$
    \IF{$v > 0$}
        \STATE $y \sim gm([1], [\mu_2], [1])$
    \ELSE
        \STATE $y \sim gm([1], [-2], [1])$
    \ENDIF
\end{algorithmic}
\end{minipage}
\vspace{-1em} 
\end{wrapfigure}
 MCMC can correctly perform inference through enumeration.\footnote{https://pyro.ai/examples/enumeration.html} 
However, when the condition depends on continuous variables, the resulting model becomes discontinuous and nondifferentiable, breaking the smoothness assumptions of gradient-based VI and affecting the convergence of MCMC. 
This behavior is clearly visible in the program on the right inset adapted from~\cite{lee2019towards}, 
where $\mu_1$ and $\mu_2$ are optimizable parameters and whose true value is $0.5$ and $1$.   Figure~\ref{plot:discontinuous_loss} shows that VI has oscillatory loss values (similarly, MCMC does not reach a value of $\hat{R} < 1.05$, indicating nonconvergence), while DeGAS optimizes successfully. 
This highlights the complementarity of DeGAS with existing methods: while it may incur longer runtimes on some benchmarks, it can also enable efficient optimization in challenging discontinuous settings, where VI and MCMC are known to struggle.


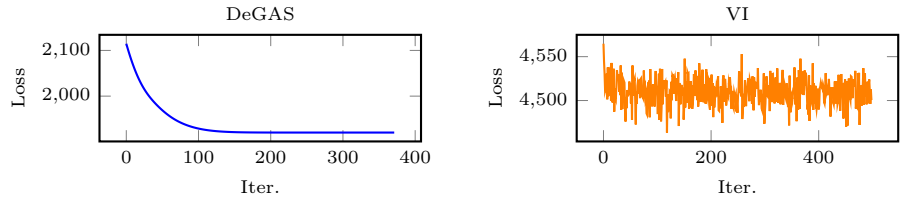
\begin{figure}[t]
\centering

\begin{minipage}{0.48\linewidth}
\begin{tikzpicture}
\begin{axis}[
    width=\linewidth,
    height=3cm,
    xlabel={Iter.},
    ylabel={Loss},
    title={DeGAS},
    thick,
    grid=none,
    label style={font=\scriptsize},
    tick label style={font=\scriptsize},
    title style={font=\scriptsize, yshift=-1ex}
]
\addplot[blue, thick] 
    table [x expr=\coordindex, y=mean, col sep=comma] 
    {csv_files/losses_test.csv};
\end{axis}
\end{tikzpicture}
\end{minipage}
\hfill
\begin{minipage}{0.48\linewidth}
\begin{tikzpicture}
\begin{axis}[
    width=\linewidth,
    height=3cm,
    xlabel={Iter.},
    ylabel={Loss},
    title={VI},
    thick,
    grid=none,
    label style={font=\scriptsize},
    tick label style={font=\scriptsize},
    title style={font=\scriptsize, yshift=-1ex}
]
\addplot[orange, thick] 
    table [x expr=\coordindex, y=mean, col sep=comma] 
    {csv_files/losses_VI_test.csv};
\end{axis}
\end{tikzpicture}
\end{minipage}

\caption{Comparison of DeGAS and VI loss curves for the program in Section 5.}
\label{plot:discontinuous_loss}
\end{figure}


\paragraph*{Parameter synthesis for cyberphysical systems.} We now consider models of cyber-physical systems (CPS) where control-flow branches depend on continuous random variables such that the previously described nondifferentiability issue arises. 
DeGAS, instead, provides an analytic form for the loss functions: we show examples of  both data likelihoods and complex reachability objectives.
We consider the following models taken from \cite{shmarov2015probreachverifiedprobabilisticdeltareachability,chaudhuri2010smooth} and fully described in Appendix \ref{app:cps}:
\begin{enumerate}
    \item \textbf{Thermostat} represents a thermostat maintaining fixed temperature by switching an heater on or off when the temperature crosses the thresholds $t_{on}$ and $t_{off}$ (optimizable parameters). The temperature evolves according to a first-order decay process with an added heating term when the heater is active and stochastic fluctuations due to noise.

    \item \textbf{Gearbox} represents a car’s transmission system. We consider a model with three gears. The controller determines when to shift from one gear to the next, allowing only consecutive transitions.
    The thresholds $s_i$ (optimizable parameters) specify the velocity at which gear $i$ is released and gear $i+1$ engages. 
    Additionally, gear shifts are not instantaneous: in our model, each transition lasts $0.3$ seconds, during which the gearbox remains disengaged. 
    The velocity grows linearly when the gear is engaged and decreases quadratically when disengaged. In both cases it is subject to stochastic fluctuations.
    
    \item \textbf{Controlled Bouncing Ball} represents a ball dropped onto a platform with a spring and damper. The spring pushes the ball back with a force determined by the compliance $C$, and the damper slows the ball according to the coefficient $R$ (both are optimizable parameters). When the ball reaches the platform, the system transitions to a reflection mode, returning to free fall once the height is positive.
\end{enumerate}
\paragraph{Parameter synthesis by data likelihoods.} 
As in the previous case studies, the goal of the optimization is to maximize the negative log-likelihood with respect to a set of $100$ trajectories generated from the models using the true parameter values.
Table~\ref{tab:cps_results} and Figure~\ref{plot:opt_by_traj} show that DeGAS consistently recovers parameter values close to the ground truth. The optimized mean trajectories (in red) accurately capture the behavior of the stochastic simulations (in grey), even when the initial trajectories (in blue) are far from the correct one. The normalized loss curves confirm convergence in all three cases.

\begin{table}[t]
\centering
\caption{Results for the CPS optimization from trajectories}
\begin{tabular}{lccccccc}
\toprule
\textbf{Case Study} & \textbf{Params} & \textbf{lr} & \textbf{Steps} & \textbf{Time} & 
\textbf{Init} & \textbf{Target} & \textbf{Result} \\
\midrule
Thermostat & $t_{ON}, t_{OFF} $ & $0.1$ & $40$ & $61.964$ & $(15,22)$ & $(17, 20)$ & $(16.78, 19.97)$ \\
Gearbox       & $s_1, s_2$ & $0.15$ & $200$ & $167.643$ & $(8, 12)$ & $(10, 20)$ & $(8.69, 19.24)$ \\
Bouncing Ball & $R, 1/C$ & $0.8$ & $100$ & $500.621$ & $(-1, 450)$ & $(7, 400)$ & $(5.27, 408.99)$ \\
\bottomrule
\end{tabular}
\label{tab:cps_results}
\end{table}

\begin{figure}[t]
    \centering
\begin{tikzpicture}
\begin{groupplot}[
    group style={group size=2 by 2, horizontal sep=0.9cm, vertical sep=1.7cm},
    width=6cm, height=3.5cm,
    xlabel={Time step},
    ylabel={},
    grid=both,
    ticklabel style={font=\scriptsize},
    label style={font=\scriptsize},
    title style={font=\small, yshift=-1pt},
    legend style={
        at={(0.5,-0.15)},
        anchor=north,
        legend columns=-1,
        font=\scriptsize,
    },
]

\newcommand{\plotcase}[4]{%
    \pgfplotstableread[col sep=comma]{#2}\rawtraj
    \pgfplotstableread[col sep=comma]{#3}\preddata
    \pgfplotstableread[col sep=comma]{#4}\initdata
    
    \pgfplotstablegetcolsof{\rawtraj}
    \pgfmathtruncatemacro{\numcols}{\pgfplotsretval-1}
    
    \foreach \r in {0,...,14} { 
        \addplot[black!30, very thin] table [
            x expr=\coordindex,
            y expr=\thisrowno{\r}
        ] {\rawtraj};
    }
    \addlegendimage{gray!60, thin}

    \addplot[very thick, red] table [x expr=\coordindex, y index=0] {\preddata};

    \addplot[very thin, blue!60] table [x expr=\coordindex, y index=0] {\initdata};

    \addplot[name path=upper, draw=none]
        table [x expr=\coordindex, y expr=\thisrowno{0}+\thisrowno{1}] {\preddata};
    \addplot[name path=lower, draw=none]
        table [x expr=\coordindex, y expr=\thisrowno{0}-\thisrowno{1}] {\preddata};
    \addplot[red!50, opacity=0.5] fill between[of=upper and lower];

    \addplot[name path=upper, draw=none]
        table [x expr=\coordindex, y expr=\thisrowno{0}+\thisrowno{1}] {\initdata};
    \addplot[name path=lower, draw=none]
        table [x expr=\coordindex, y expr=\thisrowno{0}-\thisrowno{1}] {\initdata};
    \addplot[blue!20, opacity=0.5] fill between[of=upper and lower];
}

\nextgroupplot[title={Thermostat}]
\plotcase{green!60!black}{csv_files/thermostat_orig_traj.csv}{csv_files/thermostat_opt.csv}{csv_files/thermostat_init.csv}

\nextgroupplot[title={Gearbox}]
\plotcase{green!60!black}{csv_files/gearbox_orig_traj.csv}{csv_files/gearbox_opt.csv}{csv_files/gearbox_init.csv}

\nextgroupplot[title={Controlled Bouncing Ball}]
\plotcase{green!60!black}{csv_files/bouncing_ball_orig_traj.csv}{csv_files/bouncing_ball_opt.csv}{csv_files/bouncing_ball_init.csv}

\nextgroupplot[
  title={Normalized Loss Functions},
  xlabel={Epoch},
  ylabel={Loss},
  xmin=0, xmax=100,
  legend style={at={(0.5,-0.2)}, anchor=north, legend columns=-1},
]
\pgfplotstableread[col sep=comma]{csv_files/thermostat_loss.csv}\lossA
\pgfplotstableread[col sep=comma]{csv_files/gearbox_loss.csv}\lossB
\pgfplotstableread[col sep=comma]{csv_files/bouncing_ball_loss.csv}\lossC

\addplot[thick, orange] table [x expr=\coordindex, y index=0] {\lossA};
\addplot[thick, brown] table [x expr=\coordindex, y index=0] {\lossB};
\addplot[thick, pink] table [x expr=\coordindex, y index=0] {\lossC};

\legend{T,G,B}
\end{groupplot}

\end{tikzpicture}
\caption{Mean trajectories with one standard deviation (red) obtained using optimized parameters, compared to trajectories generated with initial parameters (blue). The gray trajectories represent the samples used for optimization.}
\label{plot:opt_by_traj}
\end{figure}
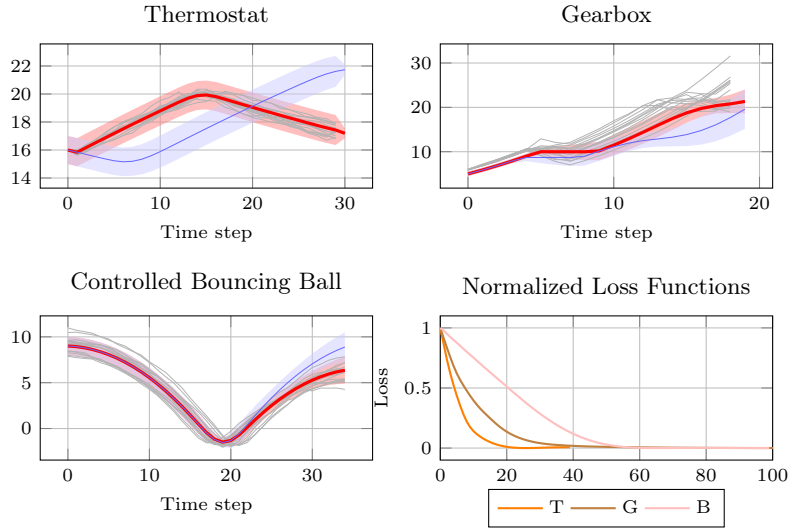

\subsubsection{Synthesis by reachability}

A complementary synthesis objective consists in optimizing the parameters of the CPS to maximize the probability of reaching specific target states or satisfying certain behavioral constraints. 
We consider three such cases. 
Thermostat: the goal is to maximize the probability of reaching the region $T(t) \in [T_{min}, T_{max}] = [19.9, 20.1]$ in mode "ON" at various time points $(\tau_1 = 0.6, \tau_2 = 1.8 \text{ and } \tau_3 = 2.4)$; 
gearbox: we maximize the probability of maintaining velocity below $16$ across the whole trajectory;  bouncing ball: we want to maximize
the probability that the ball reaches $H \geq 7$ when falling after making one bounce. The optimizable parameters are the same as the ones in the previous section.

We report the results in Table \ref{tab:cps_reach_results} and Figure \ref{fig:opt_reach}. The table reports the initial and final loss values, since the true optimal parameters are unknown. 
As shown in the figure, for all models the optimized trajectories (in red) satisfy the desired behavioral constraints, unlike the initial trajectories (in blue). 
The optimization curves further confirm convergence across all three models.

\begin{table}[t]
\centering
\caption{Results for the CPS optimization from reachability properties}
\begin{tabular}{lccccccc}
\toprule
\textbf{Case Study} & \textbf{lr} & \textbf{Steps} & \textbf{Time(s)} & 
\textbf{Init} & \textbf{Init Loss} & \textbf{Result} & \textbf{Fin. Loss}\\
\midrule
Thermostat & $0.1$ & $100$ & $2770$ & $(16.5,22)$ & $-2e^{-12}$ & $(19.04, 21.30)$ & $-5695$\\
1Gearbox & $0.5$ & $40$ & $52.642$ & $(10, 20)$ & $1.000$ & $(1.92, -1.57)$ & $0.966$\\
Bounc. Ball & $0.2$ & $50$ & $184.264$ & $(7, 200)$ & $-0.002$ & $(-3.87, 210.14)$ & $-0.979$\\
\bottomrule
\end{tabular}
\label{tab:cps_reach_results}
\end{table}

\begin{figure}[t]
    \centering
\begin{tikzpicture}
\begin{groupplot}[
    group style={group size=2 by 2, horizontal sep=0.9cm, vertical sep=1.7cm},
    width=6cm, height=3.5cm,
    xlabel={Time step},
    ylabel={},
    grid=both,
    ticklabel style={font=\scriptsize},
    label style={font=\scriptsize},
    title style={font=\small, yshift=-1pt},
    legend style={
        at={(0.5,-0.20)},
        anchor=north,
        legend columns=-1,
        font=\scriptsize,
    },
]

\newcommand{\plotcase}[3]{%
    \pgfplotstableread[col sep=comma]{#2}\preddata
    \pgfplotstableread[col sep=comma]{#3}\initdata

    \addplot[very thick, red] table [x expr=\coordindex, y index=0] {\preddata};

    \addplot[very thin, blue] table [x expr=\coordindex, y index=0] {\initdata};

    \addplot[name path=upper, draw=none]
        table [x expr=\coordindex, y expr=\thisrowno{0}+\thisrowno{1}] {\preddata};
    \addplot[name path=lower, draw=none]
        table [x expr=\coordindex, y expr=\thisrowno{0}-\thisrowno{1}] {\preddata};
    \addplot[red!50, opacity=0.5] fill between[of=upper and lower];

    \addplot[name path=upper, draw=none]
        table [x expr=\coordindex, y expr=\thisrowno{0}+\thisrowno{1}] {\initdata};
    \addplot[name path=lower, draw=none]
        table [x expr=\coordindex, y expr=\thisrowno{0}-\thisrowno{1}] {\initdata};
    \addplot[blue!20, opacity=0.5] fill between[of=upper and lower];
}

\nextgroupplot[title={Thermostat}]
\plotcase{red}{csv_files/thermostat2_opt.csv}{csv_files/thermostat2_init.csv}
\addplot[
    only marks,
    mark=*,
    mark size=2pt,
    color=gray
] coordinates {
    (6, 20)
    (18, 20)
    (24, 20)
};


\nextgroupplot[title={Gearbox}]
\plotcase{green!60!black}{csv_files/gearbox2_opt.csv}{csv_files/gearbox2_init.csv}
\addplot[dashed, gray, thick, domain = 0:20]{16};

\nextgroupplot[title={Controlled Bouncing Ball}]
\plotcase{green!60!black}{csv_files/bouncing_ball2_opt.csv}{csv_files/bouncing_ball2_init.csv}
\addplot[dashed, gray, thick, domain = 0:35]{7};

\nextgroupplot[
  title={Normalized Loss Functions},
  xlabel={Epoch},
  ylabel={Loss},
  xmin=0, xmax=100,
]
\pgfplotstableread[col sep=comma]{csv_files/thermostat2_loss.csv}\lossA
\pgfplotstableread[col sep=comma]{csv_files/gearbox2_loss.csv}\lossB
\pgfplotstableread[col sep=comma]{csv_files/bouncing_ball2_loss.csv}\lossC

\addplot[thick, orange] table [x expr=\coordindex, y index=0] {\lossA};
\addplot[thick, brown] table [x expr=\coordindex, y index=0] {\lossB};
\addplot[thick, pink] table [x expr=\coordindex, y index=0] {\lossC};

\legend{T,G,B}
\end{groupplot}
\end{tikzpicture}
\caption{Mean trajectories with one standard deviation (red) obtained by optimizing the target represented in gray are compared with initial trajectories (blue)}
\label{fig:opt_reach}
\end{figure}
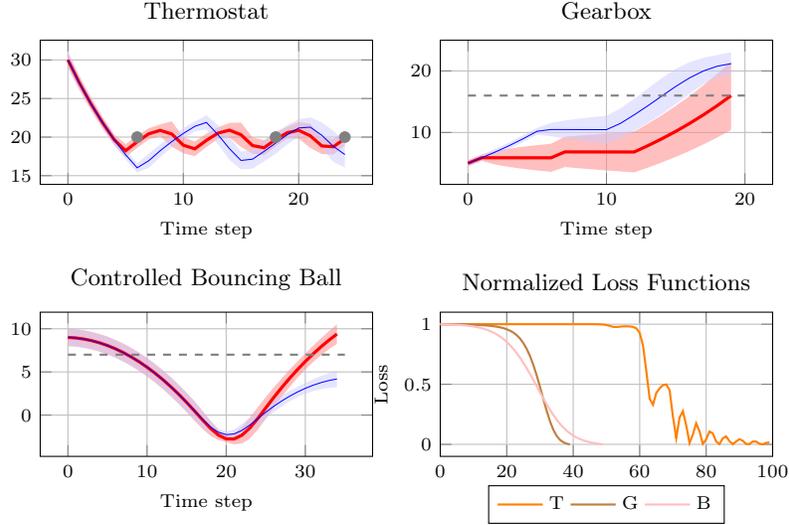


\section{Conclusions}\label{sec:conclusion}
We introduced DeGAS, a method to enable sample-free optimization of probabilistic programs via a differentiable program semantics. DeGAS smoothens SOGA, a semantics which represents a program evolution through Gaussian-mixture updates (affine maps, products, and truncated moments) that enable an analytic---but, in general, discontinuous,  representation of the posterior distribution. We demonstrated the potential of DeGAS in tackling synthesis tasks that cannot be directly solved through classic optimization methods such as variational inference and Markov chain Monte Carlo, due to the presence of discontinuities hindering convergence. 

In principle, the main ideas behind DeGAS can be applied to any semantics that (i) represents program states in a tractable form, (ii) supports closed-form (or deterministic-quadrature) updates for transforms and conditioning, and (iii) admits a smooth relaxation of discrete or measure-zero constructs whose parameter $\varepsilon \to 0$ recovers the original semantics. We conjecture that under these conditions, all statements about differentiability hold for $\varepsilon > 0$, and convergence back to the unsmoothed semantics follows by similar arguments as for SOGA.
Thus, an interesting direction for future work is to study other semantics such as probabilistic circuits, mixture of exponentials or general symbolic PPs that provide analytic or semi-analytic representations of distributions.

\section*{Data Availability Statement}
DeGAS is publicly available on Github at \url{https://github.com/frarandone/DeGAS}. This paper was accompanied by an artifact for the replication of the experimental results which can be found at: \url{https://zenodo.org/records/18197807}.

\section*{Acknowledgements}
This project has received funding from the European Union’s Horizon 2020 research and innovation programme under the Marie Skłodowska Curie grant agreement No 101205923 (MSCA Post-doctoral Fellowship EMBEr).

\noindent \raisebox{-0.2ex}{\includegraphics[height=3ex]{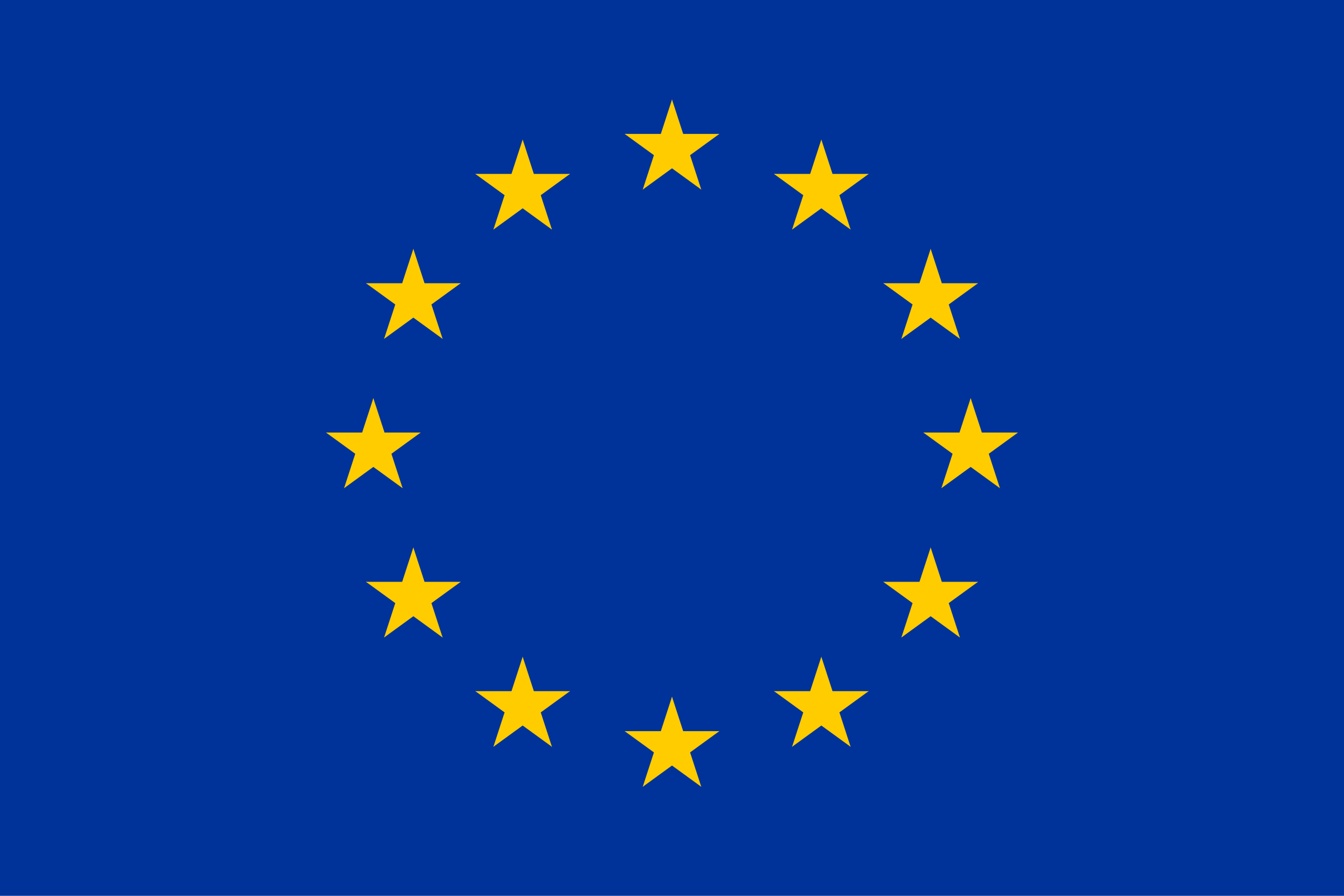}}

\bibliographystyle{splncs04}
\bibliography{bibliography}

\clearpage
\appendix
\section*{Appendix}
\section{Exact and SOGA semantics}
\label{app:soga}

\subsection{Definition of SOGA semantics}

To define the SOGA semantics, as in \cite{DBLP:journals/pacmpl/RandoneBIT24}, first the exact control-flow graph (cfg) semantics from \cite{kozen1979semantics} is introduced. 

We use the same labelling used in Section \ref{sec:diff-semantics}.

The exact semantics of a node $\node$ is a function taking as input a pair $(\p, \dist)$, and returning a pair  $(\p', \dist')$ of the same type. The only exception is the $\entrytype$ node which takes no input. We define the semantics of each node type as follows.

\begin{itemize}
	\item If $\node \colon \entrytype$, then  $\exact{\node}_\path = (1, \normal(0, \mathbb{I}_n)).$
	\item If $\node \colon \detstatetype$, then, letting $\dist'$ denote the distribution of $x[\xi \to \aexp]$,
	$$ \exact{\node}_{\path}(\p, \dist) = \begin{cases} (\pi, \dist) & \text{ if } \argfunc(\node) = \stmtSkip \\ (\p, \dist')  & \text{ if } \argfunc(\node) = \stmtAsgn{\xi}{\aexp}. \end{cases}$$ 
	\item If $\node \colon \rndstatetype$ and $\argfunc(\node) = \stmtDist{\xi}{\gmstmt}$, then $$\exact{\node}_\path(\p, \dist) = (\p, \marg{x \setminus \xi}{D} \otimes \gmg).$$
	\item If $ \node \colon \testtype$, letting $\argfunc(\node) = \bexp$:
	$$ \exact{\node}_\path(\p, \dist) = \begin{cases} 
		(\p \cdot \prob{\dist}{\bexp}, \cond{\dist}{\bexp}) & \condfunc(\succfunc(\node)) = \textit{true} \text{ and } \prob{\dist}{\bexp} \neq 0 \\
		(\p \cdot \prob{\dist}{\neg \bexp}, \cond{\dist}{\neg \bexp}) & \condfunc(\succfunc(\node)) = \textit{false} \text{ and } \prob{\dist}{\neg \bexp} \neq 0 \\		
		(0, \dist) & \text{else}. \end{cases} $$
	\item If $ \node \colon \observetype$, letting $\argfunc(\node) = \bexp$ and $I = \int_{\R^{n-1}} \density{\dist}(x[\xi \to c]) d(x\setminus \xi)$:
	$$ \exact{\node}_\path(\p, \dist) = \begin{cases}
		(\p \cdot I, \cond{\dist}{\bexp}) & \bexp \text{ is } \xi == c \\
		(\p \cdot \prob{\dist}{\bexp}, \cond{\dist}{\bexp}) & \bexp \text{ is } \xi \bowtie c \text{ and } \prob{\dist}{\bexp} \neq 0 \\
		(0, \dist) & \text{else}.
	\end{cases} $$
	\item If $\node \colon \exittype$, $\exact{\node}_\path(\p, \dist) = (\p, \dist)$.
\end{itemize}

The semantics of paths and valid programs is defined as for the differentiable semantics.


	
	
	
	

\subsection{SOGA semantics}

%
%

The SOGA semantics is then defined for each node using the moment-matching operator $\mmop$ defined in Section \ref{sec:diff-semantics}
$$ \soga{\node}_\path = \begin{cases} 
	\exact{\node}_\path & \node \colon \entrytype, \exittype \\
	(\id, \mmop) \circ \exact{\node}_\path & \text{else}.
	\end{cases} $$

The semantics of paths and programs is defined as in the previous cases.


\section{Additional proofs}

\addtocounter{theorem}{-2}

\subsection{Proof of Theorem 1}
\label{app:proofs-diff}

\begin{theorem}
	For any valid program $\stmt(\params)$, $\degas{\stmt(\params)}$ returns $(\p'(\params), \dist'(\params))$ such that $\p'(\params)$ and $\dist'(\params)$ are differentiable in $\params$.
\end{theorem}

\proof

Since the semantics of a valid program is defined as 
$$ \degas{\stmt} = \sum_{\substack{\path \in \pathset{\stmt} \\ \degas{\path} = (\p, \dist)} } \frac{ \p \cdot \dist}{\sum_{\substack{\path' \in \pathset{\stmt} \\ \degas{\path'} = (\p', \dist')}} \p'}$$
and semantics of each path is defined as $$ \degas{\path} = \degas{\node_n}_\path \circ \hdots \circ \degas{\node_0}_\path$$
we can prove for each node $\node$ that if $(\p(\params), \dist(\params))$ are differentiable in $\params$ so is $\degas{\node}_\path(\p(\params), \dist(\params), \varset)$.

Without loss of generality it is assumed that $\dist(\params)$ is a Gaussian distribution with mean $\mu(\params)$ and covariance matrix $\Sigma(\params)$, where $\mu$ and $\Sigma$ are differentiable functions of $\params$ and $\Sigma$ is non-singular for every $\params$. 
At node $\node$, $\degas{\node}_\path$ outputs a $\p'(\params)$ and $\dist'(\params)$ where $\dist'$ is possibly a GM. In this case we denote with $\mu'(\params)$, $\Sigma'(\params)$ the mean and the covariance matrix of any component of $\dist'(\theta)$.
For each node type it is proved that $\p'(\params)$, $\mu'(\params)$ and $\Sigma'(\params)$ are differentiable in $\params$ and that $\Sigma'(\params)$ is non-singular.
The general case in which $\dist(\params)$ (resp. $\dist'(\params)$) is a GM can be recovered by applying the same arguments componentwise. 

\begin{itemize}[label=\textbullet]
	
	\item $\node : \detstatetype$. If $\argfunc(\node) = \stmtSkip$ diffrentiability follows trivially from the differentiability of $\p(\params)$ and $\dist(\params)$.
	If $\arg(\node) = \stmtAsgn{\xi}{\aexp}$, then $\degas{v}_\path(\p, \dist, \varset) = (\p, \dist', \varset)$ so we only need to prove differentiability of $\dist'$. We distinguish three cases.
	
	\begin{itemize}[label=$\ast$]
		\item $\aexp$ is linear, $\xi$ is not in $\aexp$. 
		
		Let $x'$ denote the subvector $x \setminus \xi$ and consider the vector $(x', x_i, z)$. Before the assignment the vector is Gaussianly distributed with
		$$ \mu(\params) = \begin{bmatrix} \mu(\params)_{x'} \\ \mu(\params)_{\xi} \\ 0 \end{bmatrix}, \quad \Sigma(\params) = \begin{bmatrix} 
			\Sigma_{x',x'} & \Sigma_{x',\xi} & 0 \\ 
			\Sigma_{\xi,x'} & \Sigma_{\xi,\xi} & 0 \\
			0 & 0 & \smoothEps^2 \end{bmatrix}. $$
		
		Let $\alpha(\params)$ be the row vector representing the linear combination in $e$ and $\alpha_0(\params)$ the zero-th order term. 
		Both $\alpha(\params)$ and $\alpha_0(\params)$ are differentiable in $\params$ since the coefficient are either constant or parameters themself. 
		Let $A(\theta)$ be the blockmatrix
		$$ A(\params) = \begin{bmatrix}
			\mathbb{I}_{n-1} & 0 & 0 \\
			\alpha(\params) & 0 & 1 \\
			0 & 0 & 1 \end{bmatrix}. $$
		
		After the assignment $\stmtAsgn{\xi}{e+z}$ the distribution over program variables is Gaussianly distributed with mean $\mu^*(\params)$ and covariance matrix $\Sigma^*(\params)$ given by (we suppress dependence on $\params$ at the r.h.s. for the sake of notation):
		$$ \mu^*(\params) = \begin{bmatrix} \mu_{x'} \\ \alpha \mu + \alpha_0 \\ 0 \end{bmatrix},  \Sigma^*(\params) = A\Sigma A^T = \begin{bmatrix} 
			\Sigma_{x',x'} & \Sigma_{x',\xi}\alpha^T & 0 \\
			\alpha\Sigma_{\xi,x'} & \alpha\Sigma_{x',x'}\alpha^T + \smoothEps^2 & \smoothEps^2 \\
			0 & 0 & \smoothEps^2 \end{bmatrix}. $$
		When considering the marginal over $x$, by properties of the Gaussians \cite{bishop2006pattern}, we just need to consider the submatrices indexed by $x = (x', \xi)$.
		Since $\mu, \Sigma, \alpha$ and $\alpha_0$ depend differentiably on $\params$ so do $\mu^*_{x}$ and $\Sigma^*_{x,x}$. 
		However, to prove that the resulting pdf is differentiable in $\params$ we still need to prove that $\Sigma^*_{x,x}$ is non-singular.
		This can be done by computing the determinant using Schur complement \cite{boyd2004convex}:
		$$ | \Sigma^*_{x,x} | = \alpha\Sigma_{x',x'}\alpha^T + \smoothEps^2 - \alpha\Sigma_{\xi,x'}\Sigma_{x',x'}^{-1}\Sigma_{x',\xi}\alpha^T = \smoothEps^2 > 0.$$
		
		\item $\aexp$ is linear, $\xi$ in $\aexp$. For this case the computations are analogous to the ones carried out above, except that the variable $z$ is not needed.
		
		\item $\aexp$ is non-linear, i.e. $\stmtAsgn{\xi}{\xj\xk}$. The differentiability with respect to $\params$ of the mean and the covariance matrix after the assignment, follows from the fact that this quantities can be computed via Isserlis's theorem \cite{wick1950evaluation} as polynomial functions of $\mu(\params)$ and $\Sigma(\params)$. Again, to prove differentiability of the pdf we need to prove that the new covariance matrix $\Sigma^*(\params)$ is non-singular. Let $x'$ denote the sub-vector $x' = x \setminus \{\xi, \xj, \xk\}$. We prove that the covariance matrix of $\tilde{x} = (x', \xj\xk)$ is non-singular. For the vector $(\xj, \xk, \xk\xk)$ the proof is analogous. These two facts together prove that $\Sigma^*$ is non-singular. We proceed by contraddiction, supposing that the covariance matrix $\tilde{\Sigma}$ of $\tilde{x}$ is singular. Then exists a non-zero vector $c$ such that $c^T\tilde{\Sigma}c=0$. This means that the variance of $c^T(\tilde{x}-\mathbb{E}[\tilde{x}])$ is $0$. But this means that $c^T(\tilde{x}-\mathbb{E}[\tilde{x}])$ is deterministic, and in particular equal to its mean, $0$. Therefore, $c$ can be partioned into $a$, $b$ such that:
		\begin{equation} \label{eq:proof1}
			a(\xj\xk - \mathbb{E}[\xj\xk]) + b^T(x' - \mathbb{E}[x']) = 0.
		\end{equation}
		We now use a known property of Gaussian vectors: the conditional expectation of a component with respect to others is an affine function \cite{bishop2006pattern}. This means that $\mathbb{E}[x' | \xj, \xk] = \alpha \xj + \beta \xk + \gamma$ for some vectors $\alpha, \beta, \gamma$. Therefore, conditioning \eqref{eq:proof1} to $\xj$ and $\xk$ and taking the mean we find that exist constants $\alpha, \beta, \gamma, \delta$ such that:
		$$ \alpha \xj\xk + \beta\xj + \gamma\xk + \delta = 0.$$
		If $\alpha = 0$ the equality reduces to $\beta \xj + \gamma \xk + \delta = 0$, which would imply a degeneracy in the original covariance matrix $\Sigma$, causing a contraddiction.
		If $\alpha \neq 0$ we can rewrite $\xj = \frac{ -\gamma \xk - \delta}{\alpha \xk + \beta}$.
		Since $\xk$ is non degenerate by hypothesis, the denominator is different than zero a.e.. 
		But by the property of Gaussian vectors, $\mathbb{E}[\xj | \xk]$ must be an affine function, which is violated by the fact that $\xj$ is a rational function of $\xk$. Therefore $\Sigma^*$ cannot be singular.
	\end{itemize}
	
	\item $\node \colon \rndstatetype$. In this case the differentiability is guaranteed by the differentiability of $\dist(\params)$ and of $\gmstmt$, which has never components with $0$ std.
	
	\item $\node \colon \observetype$. We distringuish the three cases.
	\begin{itemize}[label=$\ast$]
		
		\item $\bexp$ is $\xi == c$ and $\xi \not \in \varset$. Differentiability follows from the differentiability of the pdfs of $\normal(c, \smoothEps)$ and of $\cond{\dist}{\bexp}$, guaranteed by the fact that its moments are polynomials in $\mu(\params)$ and $\Sigma(\params)$. 
		In addition, if $\Sigma(\params)$ is non-degenerate so is the covariance matrix of $\cond{\dist}{\bexp}$.
		
		\item $\bexp$ is $\xi \bowtie c$ in all the other cases. 
		In this case the predicate is transformed into $\bexp' = \smoothOpBool(\xi \bowtie c, \varset)$. 
		Since $\dist$ is non-degenerate and all the conditions defined by $\smoothOpBool(\bexp, \varset)$ have positive measure $\prob{\dist}{\bexp'}$ is differentiable and different than 0 (it is computed using the cdf of $\dist$ which is a differentiable function).
		The moments of the truncated distribution $\cond{dist}{\bexp'}$ are computed using the formulas in \cite{kan2017moments} as functions of $\mu(\params)$ and $\Sigma(\params)$ and of the extrema of the truncation ($-\infty$ and $+\infty$ for all variables except for $\xi$, for which they involve $c$). 
		Since $\dist$ is non-degenerate, the formulas in \cite{kan2017moments} are differentiable in their arguments. 
		Therefore, the moments of $\cond{dist}{\bexp'}$ are be differentiable in $\params$.
		In addition, we need to prove that the new covariance matrix will be non-singular, but this follows from the fact that the initial covariance matrix is non-singular and the truncations are made on sets of positive measure. 
	\end{itemize}
	
	\item $\node \colon \testtype$. 
	We let $\argfunc(\node) = \bexp$ and $\bexp' = \smoothOpBool(\bexp, \varset)$. 
	The computation of the probabilities $\prob{\dist}{b'}$ (resp. $\neg b'$) and of the moments of the conditioned distributions $\cond{\dist}{b'}$ (resp. $\neg b'$) uses the formulas from \cite{kan2017moments}, as in the case of $\observetype$ nodes. 
	Therefore both the probabilities and the truncated moments are differentiable in $\params$
	
	\item $\node \colon \entrytype$ and $\node \colon \exittype$ the conclusion follows trivially. 
\end{itemize}

\subsection{Proof of Theorem 2}
\label{app:proofs-conv}

\begin{theorem}
	Suppose $\stmt$ is a valid program and $\degas{\stmt}$ $\smoothOpBool$ is applied consinstently when computing $\degas{\stmt}$. Then $\lim_{\smoothEps \to 0} \degas{\stmt} = \soga{\stmt}$ (where convergence for the distributions is intended in the sense of \emph{convergence in distribution}).
\end{theorem}

\begin{proof}
	Recall that both for $\degas{\cdot}$ and $\soga{\cdot}$ the semantics of $\stmt$ is defined as a mixture of the semantics of the paths in $\pathset{\stmt}$. Since we are only considering finite programs, the convergence of the program semantics is implied by the convergence of the path semantics. 
	First, we only consider paths $\path \in \pathset{\stmt}$ such that $\soga{\path} = (\p, \dist)$ with $\p > 0$. 
	Consider one such path with differentiable semantics $\degas{\path} = \degas{\node_n}_\path \circ \hdots \circ \degas{\node_0}_\path$. We prove that this converges to $\soga{\path} = \soga{\node_n}_\path \circ \hdots \circ \soga{\node_0}_\path$ by induction on $n$.
	
	For $n=1$ we only have an $\entrytype$ and an $\exittype$ node, so the conclusion follows trivially. 
	
	Suppose that up to node $\node_{i-1}$ one has $\degas{\node_{i-1}}_\path \circ \hdots \circ \degas{\node_0} = (\peps, \deps, \varset)$ and $\soga{\node_{i-1}}_\path \circ \hdots \circ \soga{\node_0}_\path = (\p, \dist)$ with $\lim_{\smoothEps \to 0} (\peps, \deps) = (\p, \dist)$. We prove that $\degas{\vi}_\path(\peps, \deps, V) = (\peps', \deps', \varset')$ such that $\lim_{\smoothEps \to 0} (\peps', \deps') = (\p', \dist') = \soga{\vi}_\path(\p, \dist)$.
	Without loss of generality we assume that $\deps$ (resp. $\dist$) is a Gaussian with mean $\meps$ (resp. $\mu$) and covariance matrix $\sigmaeps$ (resp. $\Sigma$). 
	When acting on a Gaussian $\deps$, $\degas{\node}$  and $\soga{\node}$ can in general yield a GM with a fixed number of components, which is the semantics in both semantics, for all values of $\smoothEps$. 
	In this case We will denote with $\meps'$ (resp. $\mu$) and $\sigmaeps'$ (resp. $\Sigma'$) the mean and the covariance matrix of any component of $\deps'$ (resp. $\dist'$), assuming to choose corresponding components in each distributions. 
	The proof can be generalized to GMs by  applying the same arguments componentwise.
	Observe that if $\deps'$ and $\dist$ are both Gaussians it sufficient to prove that $\meps' \to \mu$ and $\sigmaeps' \to \Sigma$, since, by convergence of the characteristic function, this implies convergence in distribution.
	
	We consider the possible type of $\vi$ separately.
	
	\begin{itemize}
		\item $\node_i \colon \detstatetype$. If $\argfunc(\node) = \stmtSkip$, conclusion follows trivially. If $\argfunc(\node) = \stmtAsgn{\xi}{\aexp}$, we only need to prove convergence of $\deps'$ to $\dist'$. Suppose $\aexp$ is linear with $\alpha$ vector of coefficients and $\alpha_0$ zero-th order term. Let $x'$ denote the vector $x \setminus \xi$, so that $x$ can be represented as $(x', \xi)$. Then:
		$$ \meps' = \begin{bmatrix} [\meps]_{x'} \\ \alpha \meps + \alpha_0 \end{bmatrix}$$
		$$  \sigmaeps' = A\sigmaeps A^T =
		\begin{cases} 
			\begin{bmatrix} 
				[\sigmaeps]_{x',x'} & [\sigmaeps]_{x',\xi}\alpha^T \\
				\alpha[\sigmaeps]_{\xi,x'} & \alpha[\sigmaeps]_{x',x'}\alpha^T + \smoothEps^2
			\end{bmatrix} & \xi \text{ not in } \aexp \\
			\begin{bmatrix} 
				[\sigmaeps]_{x',x'} & [\sigmaeps]_{x',\xi}\alpha^T \\
				\alpha[\sigmaeps]_{\xi,x'} & \alpha[\sigmaeps]_{x',x'}\alpha^T
			\end{bmatrix} & \xi \text{ in } \aexp
		\end{cases}
		$$ 
		For $\smoothEps \to 0$ it is trivial to check that $\meps'$ and $\sigmaeps'$ converge to $\mu'$, $\Sigma'$.
		If $\aexp$ is non linear, $\meps'$ and $\sigmaeps'$ are obtained using Isserlis' theorem \cite{wick1950evaluation} as polynomial functions of $\meps$ and $\sigmaeps$. Since polynomials are continuous, the limit is preserved.
		
		\item $\node_i \colon \rndstatetype.$ Again, we only need to prove convergence for $\deps$ but this follows immediately from convergence of $\gm_{\pi, \mu, \sigma'}$ to $G_{\pi, \mu, \sigma}$ (recall $\sigma_j' = \sigma_j$ if $\sigma_j>0$ and $\smoothEps$ else).
		
		\item $\node_i \colon \observetype$. 
		If $\argfunc(\node) = \xi == c$ and $\xi \not \in \varset$ convergence follows immediately from $\lim_{\smoothEps \to 0} \normal(c, \smoothEps) = \normal(c, 0)$. 
		Let $\bexp$ be $\xi \bowtie c$ in one of the remaining cases.
		We consider two different scenarios.
		If $\Sigma_{\xi, \xi} \neq 0$ we have that$\prob{\dist}{\partial\bexp} = 0 $ where $\partial\bexp$ denotes the border of the set $\bexp$. In this case $\bexp$ is called a continuity set for $\dist$. Moreover, since $\varset$ only contains variables that would be discrete in the original program $\xi \not \in \varset$ and $\smoothOpBool(\bexp, \varset) = \bexp$. Therefore convergence reduces to proving $\prob{\deps}{\bexp} \to \prob{\dist}{\bexp}$ and $\cond{(\deps)}{\bexp} \to \cond{\dist}{\bexp}$. Since $\bexp$ is a continuity set for $\dist$ both limits are guaranteed by Portmanteau lemma \cite{billingsley2013convergence}. 
		If $\Sigma_{\xi, \xi} = 0$ then $\marg{\xi}{D}$ has a point mass at $\xi^0$. Moreover, since we are working under the hypothesis that $\degas{\path} = (\p_\path, \dist_\path)$ with $\p_\path > 0$, then $\prob{\dist}{\bexp} > 0$. Therefore, $\xi^0 \in (\bexp)_i$, where $(\bexp)_i$ is the projection of $\bexp$ on the $i$-th coordinate. Therefore $\prob{\dist}{\bexp} = 1$ and $\cond{\dist}{\bexp} = \dist$. For $\deps$ we have $(\sigmaeps)_{\xi, \xi} \to 0$ and $(\meps)_{\xi} \to \xi^0$. Now consider let $\beps = \smoothOpBool(b, V)$ and $(\beps)_i$ be the projection of $\beps$ to the $i$-th coordinate. Then $\xi^0 \in (\beps)_i$. Moreover, since $\smoothOpBool$ is applied consistently $\smoothDelta(\smoothEps) / \sqrt{\sigmaeps} \to +\infty$ and this implies $\lim_{\smoothEps \to 0} \prob{\marg{\xi}{\deps}}{(\beps)_i} = 1$. But then $\lim_{\smoothEps \to 0} \prob{\deps}{\beps} = 1 = \prob{\dist}{\bexp}$. For the distributions we can prove the limit using the metric induced by weak convergence, called the Levi Prokhorov metric and denoted as $\lpdist(\cdot, \cdot)$. Then, 
		$$ \lpdist(\cond{(\deps)}{\beps}, \dist) \le \lpdist(\cond{(\deps)}{\beps}, \deps) + \lpdist(\deps, \dist).$$
		The first term at the r.h.s. tends to zero because $\smoothDelta(\smoothEps) / \sqrt{\sigmaeps} \to +\infty$, while the second term tends to zero by hypothesis.
		
		\item $\node_i \colon \testtype$. In this case the proof works analogously as in the $\observetype$ case.
	\end{itemize}
	We have now proved that $\degas{\path} \to \soga{\path}$ for paths with positive probabilities. For paths such that $\soga{\path} = (0, \dist)$ we can just prove that $\peps \to 0$. Let $\node_i$ be the first node in the path such that $\soga{\node_i}(\p, \dist) = (\p', \dist')$ with $\p' = 0$. Then up to this node the differentiable semantics output triples $(\peps, \deps, \varset)$ such that $\peps \to \p$, $\deps \to \dist$ and we want to prove $\peps' \to 0$. We know that $\node_i$ must either be an $\observetype$ or a $\testtype$ node. Suppose is an $\observetype$ node (the other case can be treate in the same way). Let $\arg(\node_i) = \bexp = \xi \bowtie c$. We can rule out the case in which $\bowtie$ is $==$ and $\xi \not \in \varset$ as this case can never yield probability 0. In the other cases we can distinguish again two scenarios. If $\Sigma_{\xi, \xi} \neq 0$ then $\bexp$ is a continuity set for $\dist$, $\smoothOpBool(\exp, \varset) = \bexp$ and convergence is guaranteed by the Portmanteau lemma. If $\Sigma_{\xi, \xi} = 0$, the marginal $\marg{\xi}{\dist}$ is a point mass $\xi^0 \not \in (\bexp)_i$.  The $\deps$ instead are Gaussians concentrating their probability mass around $\xi^0$. Exploiting again the consistency of the operator $\smoothOpBool$ we can again prove $\prob{\deps}{\beps} \to 0$.	
\end{proof}
\section{DeGAS}
\label{sec:degas}

In this section we give an overview of the tool \degastool that implements the differentiable semantics presented earlier.
DeGAS is implemented in PyTorch, and can be used programmatically from Python.
A prototype implementation can be found in \url{url-anonymized}.

\begin{algorithm}
	\caption{DeGAS($\textit{input\_file}, \textit{params\_dict}, \textit{loss\_func}, \smoothEps, \textit{steps}$):}
	\label{alg:degas}
	\begin{algorithmic}[1]
		\STATE $\textit{cfg} = \texttt{smooth\_cfg}(\textit{input\_file}, \smoothEps)$
		\STATE $\textit{optimizer} = \texttt{import\_optimizer\_from\_torch()}$
		\FOR{$\texttt{ \_ in range}(\textit{steps})$}
			\STATE $\textit{optimizer.zero\_grad()}$  
			\STATE $\textit{current\_dist} = \texttt{compute\_degas}(\textit{cfg}, \textit{params\_dict})$
			\STATE $\textit{loss} = \textit{loss\_func}(\textit{current\_dist})$
			\STATE $\textit{loss.backward()}$
			\STATE $\textit{optimizer.step()}$
		\ENDFOR
		\RETURN $\textit{params\_dict}$
	\end{algorithmic}
\end{algorithm}

The algorithm implemented by \degastool is outlined in Algorithm~\ref{alg:degas}. It takes as input a $\texttt{.soga}$ file $\textit{input\_file}$, which encodes a program in the syntax of Section~\ref{sec:syntax}; a dictionary $\textit{params\_dict}$ whose keys are parameter names (also appearing in the program, prefixed by an underscore) and whose values are PyTorch tensors with $\texttt{requires\_grad} = \textit{True}$; a loss function $\textit{loss\_func}$ mapping distributions to real values; the smoothing parameter $\smoothEps > 0$; and the number of optimization steps $\textit{steps}$. Where required, constraints on individual parameters can be specified directly in the optimization loop by clamping them to their respective interval domain.

The algorithm begins by constructing the control-flow graph (cfg) of the program via $\texttt{smooth\_cfg}$. This step also prepares the information needed for computing the differentiable semantics: for instance, assignments $\stmtAsgn{x}{c}$ are transformed into $\stmtAsgn{x}{c + gm([1], [0], [\smoothEps])}$, and analogous rewritings are applied to other instructions. Moreover, the set of smoothed variables $\varset$ is also pre-computed at each node. The resulting object $\textit{cfg}$ encodes the program’s cfg, enriched with all information required to evaluate $\degas{\cdot}$.

After constructing the cfg, a gradient-based optimizer from PyTorch is initialized (line~2) and the optimization loop begins (line~3). At each step, the posterior distribution is computed by $\texttt{compute\_degas}$ (line~5), the loss is evaluated (line~6), and the parameters in $\textit{params\_dict}$ are updated via backpropagation (lines~7–8). After all iterations, the optimized parameters are returned.

\paragraph{Computational cost of DeGAS.}
In computing the semantics of a program, DeGAS shares the same cfg structure of SOGA. 
The operations introduced to guarantee differentiability do not alter the asymptotic computational cost, since computation of the program differentiable semantics involves the same elementary operations and function calls described for SOGA \cite{DBLP:journals/pacmpl/RandoneBIT24}.
The overall complexity is therefore still $O(|V|2^{|T|}n^4)$,
where $|V|$ denotes the number of nodes in the cfg, $|T|$ denotes the number of $\testtype$ and $\observetype$ nodes and $n$ is the number of variables in the program, assuming that all the distributions used in $\rndstatetype$ nodes have been pushed to the initial distribution. 

Analogously to SOGA, a pruning strategy can be introduced in DeGAS to limit the number of mixture components and control the computational overhead. However, in the case studies presented in this paper this was not necessary.
\section{Experimental evaluation details}

\subsection{Convergence Experiments}
\label{app:conv}
    To show how the convergence works in practice, we consider the following three programs:
    \begin{align*}
        \stmt_1 &\colon \stmtAsgn{x}{0}; \stmtObserve{x \ge 0}; \\
        \stmt_2 &\colon \stmtAsgn{x}{0}; \stmtObserve{x > 0}; \\
        \stmt_3 &\colon \stmtDist{x}{gm([0.5, 0.5], [0, 1], [0, 0])}; \stmtObserve{x == 0};
    \end{align*}

In the following table we report the result of $\p$, $\mu$ and $\sigma$ when $\smoothEps$ decreases and $\smoothOpBool$ is applied consistently, choosing $\smoothDelta = \sqrt{\smoothEps}$. For $\smoothEps = 10^{-4}$ the values of all three outputs are equal to that of SOGA up to the third decimal digit.

\begin{table}[ht]
\centering
\small
\setlength{\tabcolsep}{3pt}
\renewcommand{\arraystretch}{0.9}
\begin{tabular}{|c|ccc|ccc|ccc|} 
    \hline & \multicolumn{3}{c|}{$P_1$} & \multicolumn{3}{c|}{$P_2$} & \multicolumn{3}{c|}{$P_3$} \\ 
    \cline{2-10} 
    & $\p$ & $\mu$ & $\sigma$ & $\p$ & $\mu$ & $\sigma$ & $\p$ & $\mu$ & $\sigma$ \\ 
    \hline 
    $\smoothEps=10^{-1}$ & 0.9992 & 0.0002 & 0.0995 & 0.0008 & 0.3434 & 0.0256 & 0.4992 & 0.0000 & 0.0991 \\
    $\smoothEps=10^{-2}$ & 1.0000 & 0.0000 & 0.0100 & 0.0000 & 0.0000 & 0.0100 & 0.5000 & 0.0000 & 0.0100 \\ 
    $\smoothEps=10^{-3}$ & 1.0000 & 0.0000 & 0.0010 & 0.0000 & 0.0000 & 0.0010 & 0.5000 & 0.0000 & 0.0010 \\ 
    $\smoothEps=10^{-4}$ & 1.0000 & 0.0000 & 0.0001 & 0.0000 & 0.0000 & 0.0001 & 0.5000 & 0.0000 & 0.0001 \\ 
    SOGA & 1.0000 & 0.0000 & 0.0000 & 0.0000 & 0.0000 & 0.0000 & 0.5000 & 0.0000 & 0.0000 \\ \hline \end{tabular}
\caption{Results for DeGAS and SOGA as the value of $\smoothEps$ decreases.}
\label{tab:degassoga}
\end{table}

By testing DeGAS on the synthesis of the parameters of probabilistic programs varying the value of $\epsilon$ we obtain consistent results except, in some cases highlighted in gray, for $\epsilon = 0.1$. These programs are more complex, involving multiple variables and nested conditional branches.

\begin{table}[H]
\centering
\begin{tabular}{c|cccc}
\toprule
\textbf{Model} & $\boldsymbol{\epsilon = 0.1}$ & $\boldsymbol{\epsilon = 0.01}$ & $\boldsymbol{\epsilon = 0.001}$ & $\boldsymbol{\epsilon = 0.0001}$ \\
\midrule
Bernoulli  & $0.001$ & $0.001$ & $0.001$ & $0.001$ \\
Burglary  & $1.6857$ & $1.6857$ & $1.6857$ & $1.6857$ \\
\rowcolor{gray!20}
ClickGraph  & $0.375$ & $0.085$ & $0.086$ & $0.086$ \\
\rowcolor{gray!20}
ClinicalTrial  & $3.463$ & $0.659$ & $0.657$ & $0.657$ \\
\rowcolor{gray!20}
Grass  & $0.695$ & $0.036$ & $0.036$ & $0.036$ \\
MurderMistery  & $0.203$ & $0.203$ & $0.203$ & $0.203$ \\
SurveyUnbiased  & $0.008$ & $0.008$ & $0.008$ & $0.008$ \\
TrueSkills  & $0.003$ & $0.003$ & $0.003$ & $0.003$ \\
TwoCoins & $0.688$ & $0.688$ & $0.688$ & $0.688$ \\
AlterMu & $0.312$ & $0.312$ & $0.312$ & $0.312$ \\
AlterMu2 & $0.096$ & $0.096$ & $0.096$ & $0.096$ \\
NormalMixtures & $0.386$ & $0.386$ & $0.386$ & $0.386$ \\
\bottomrule
\end{tabular}
\label{tab:epsilon-results}
\caption{Performance on the programs parameters synthesis varying the value of $\epsilon$
}
\end{table}

\subsection{PID}
\label{app:pid}
The Proportional Integral Derivative (PID) is described by the following dynamics.

\begin{align*}
&d_k = \pi - \theta_k \\
&T_k = s_0 d_k + s_1 \omega_k + s_2 I_k \\
&I_{k+1} = 0.9 I_k + d_k \Delta t\\
&\omega_{k+1} = \omega_k + \frac{\Delta t}{10} T_k + \mathcal{N}(0, \sigma_\omega^2) \\
&\theta_{k+1} = \theta_k + \frac{\Delta t}{2} (\omega_k + \omega_{k+1}) + \mathcal{N}(0, \sigma_\theta^2)
\end{align*}

The result of DeGAS optimization is shown in Figure \ref{fig:PID_opt_reach}.

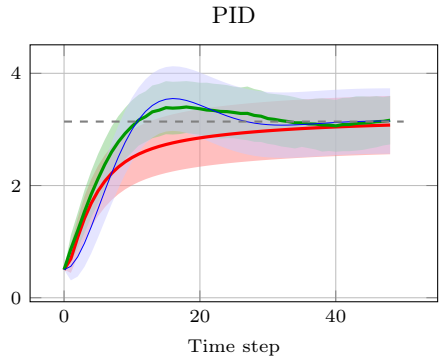
\begin{figure}[H]
    \centering
\begin{tikzpicture}
\begin{groupplot}[
    group style={group size=2 by 2, horizontal sep=0.9cm, vertical sep=1.7cm},
    width=7cm, height=5cm,
    xlabel={Time step},
    ylabel={},
    grid=both,
    ticklabel style={font=\scriptsize},
    label style={font=\scriptsize},
    title style={font=\small, yshift=-1pt},
    legend style={
        at={(0.5,-0.20)},
        anchor=north,
        legend columns=-1,
        font=\scriptsize,
    },
]
\newcommand{\plotcase}[4]{%
    \pgfplotstableread[col sep=comma]{#2}\preddata
    \pgfplotstableread[col sep=comma]{#3}\initdata
    \pgfplotstableread[col sep=comma]{#4}\VIdata

    \addplot[very thick, red] table [x expr=\coordindex, y index=0] {\preddata};

    \addplot[very thick, green!60!black ] table [x expr=\coordindex, y index=0] {\VIdata};

    \addplot[very thin, blue] table [x expr=\coordindex, y index=0] {\initdata};

    \addplot[name path=upper, draw=none]
        table [x expr=\coordindex, y expr=\thisrowno{0}+\thisrowno{1}] {\preddata};
    \addplot[name path=lower, draw=none]
        table [x expr=\coordindex, y expr=\thisrowno{0}-\thisrowno{1}] {\preddata};
    \addplot[red!50, opacity=0.5] fill between[of=upper and lower];

    \addplot[name path=upper, draw=none]
        table [x expr=\coordindex, y expr=\thisrowno{0}+\thisrowno{1}] {\VIdata};
    \addplot[name path=lower, draw=none]
        table [x expr=\coordindex, y expr=\thisrowno{0}-\thisrowno{1}] {\VIdata};
    \addplot[green!60!black!50, opacity=0.5] fill between[of=upper and lower];

    \addplot[name path=upper, draw=none]
        table [x expr=\coordindex, y expr=\thisrowno{0}+\thisrowno{1}] {\initdata};
    \addplot[name path=lower, draw=none]
        table [x expr=\coordindex, y expr=\thisrowno{0}-\thisrowno{1}] {\initdata};
    \addplot[blue!20, opacity=0.5] fill between[of=upper and lower];
}
\nextgroupplot[title={PID}]
\plotcase{red}{csv_files/PID_opt.csv}{csv_files/PID_init.csv}
{csv_files/PID_VI.csv}
\addplot[dashed, gray, thick, domain = 0:50]{3.14};

\end{groupplot}
\end{tikzpicture}
\caption{PID optimization from ideal trajectories keeping a stable dynamic with $\theta = 3.14$ leads to the red trajectory, starting from the blue one. In green the trajectory obtained by VI.}
\label{fig:PID_opt_reach}
\end{figure}

\subsection{Hyperparamters Program Optimization}
\label{app:hyper}

In Table~\ref{tab:hyper}, we report the hyperparameters selected through tuning during the experiments on parameter synthesis of probabilistic programs via data likelihood optimization.

\begin{table}[H]
\centering
\begin{tabular}{l|cc|cccc|cc}
\toprule
\textbf{Model} &
\multicolumn{2}{c|}{\textbf{VI}} & 
\multicolumn{4}{c|}{\textbf{MCMC}} & 
\multicolumn{2}{c}{\textbf{DeGAS}} \\
\cmidrule(lr){2-3} \cmidrule(lr){4-7} \cmidrule(lr){8-9}
 & \textbf{LR} & \textbf{Steps}
  & \textbf{Samples} & \textbf{Warmup} & \textbf{R-hat} & \textbf{N eff} 
 & \textbf{LR} & \textbf{Steps} \\
\midrule
Bernoulli & $0.05$ & $247$ & $500$ & $50$ & $1.005$  & $802$ & $0.01$ & $248$ \\
Burglary  & $0.05$ & $1000$ & $500$ & $50$ & $1.003$ & $1008$  & $0.005$ & $295$ \\
ClickGraph & $0.001$ & $1000$ & $500$ & $50$ & $1.018$ & $333$ & $0.2$ & $216$\\
ClinicalTrial  & $0.2$ & $1000$ & $500$ & $50$ & $1.021$ & $146$ &  $0.001$ & $500$\\
Grass  & $0.01$ & $1000$ & $500$ & $50$ & $1.004$  & $1444$ &  $0.1$ & $500$\\
MurderMistery  & $0.001$ & $1000$ & $500$ & $50$  & $1.001$ & $1036$ & $0.01$ & $239$\\
SurveyUnbiased  & $0.01$ & $675$ & $500$ & $50$ & $1.001$ & $1382$ & $0.01$ & $500$\\
TrueSkills  & $0.2$ & $1000$ & t.o. & t.o. & t.o. & t.o. & $0.2$ & $367$\\
TwoCoins  & $0.05$ & $1000$ & $500$ & $50$ & $1.015$ & $279$  & $0.001$ & $337$\\
AlterMu & $0.0005$ & $1000$ & t.o. & t.o.  & t.o. & t.o.  & $0.05$ & $235$\\
AlterMu2  & $0.001$ & $1000$ & $1500$ & $50$  & $1.040$ & $81$ & $0.2$ & $228$\\
NormalMixtures & $0.0005$ & $1000$ & $500$ & $50$ & $1.018$& $457$ & $0.01$ & $500$\\
PID & $0.05$ & $1000$ & n.c. & n.c. & n.c. & n.c. & $0.2$ & $1000$ \\

\bottomrule
\end{tabular}
\caption{Tuned hyperparameters (see Section \ref{sec:experiments}) used for the synthesis of probabilistic program parameters.}
\label{tab:hyper}
\end{table}

\subsection{CPS}
\label{app:cps}

We illustrate the three CPSs considered in our study, showing their mode-dependent dynamics, switching conditions, and key model parameters.

\begin{figure*}[h!]
\centering

\begin{minipage}{0.45\textwidth}
\centering
\adjustbox{scale=0.90, valign=t}{
\begin{tikzpicture}[
    node distance=1cm,
    every node/.style={font=\scriptsize, align=center},
    state/.style={rectangle, rounded corners, draw=black, thick,
                  text width=3.5cm, minimum height=1.5cm, inner sep=1.5pt, fill=gray!10},
    trans/.style={-{Stealth[length=1.5mm,width=0.9mm]}, very thick}
]
\node[state] (on) {\textbf{Heater ON}\\[1pt]$\frac{dT}{dt}=-kT+h+\mathcal{N}(0,\epsilon)$};
\node[state, below=of on] (off) {\textbf{Heater OFF}\\[1pt]$\frac{dT}{dt}=-kT+\mathcal{N}(0,\epsilon)$};

\draw[trans] (on.south east) .. controls +(0.5,-0.2) and +(0.5,0.2) .. node[right, text width=1.5cm, align=left] {$T>t_{off}$} (off.east);
\draw[trans] (off.north west) .. controls +(-0.5,0.2) and +(-0.5,-0.2) .. node[left, text width=1.5cm, align=right] {$T<t_{on}$} (on.west);
\end{tikzpicture}
}
\end{minipage}\hfill
\begin{minipage}{0.41\textwidth}
\caption*{\textbf{Thermostat} For the optimization with respect to trajectories, we consider $k = 0.01$, $h = 0.5$, inital $T = 16$.
In the reachability optimization (with $k = 1$, $h = 30$, and inital $T \sim \normal(30,1)$), under the assumption of independence between the events at different time steps, the loss function becomes the product of the CDF inside the region of interest multiplied by the PDF of being in mode "ON" both calculated at the three time steps. 
   }
\end{minipage}

\end{figure*}

\begin{figure*}

\begin{minipage}{0.45\textwidth}
\centering
\adjustbox{scale=0.90, valign=t}{
\begin{tikzpicture}[
    every node/.style={font=\scriptsize, align=center},
    state/.style={rectangle, rounded corners, draw=black, thick,
                  fill=gray!10, text width=3.5cm, minimum height=1.5cm, inner sep=1.5pt},
    trans/.style={-{Stealth[length=1.5mm,width=0.9mm]}, very thick}
]
\node[state] (active) at (0,1.0) {\textbf{Active (gear=i)}\\[1pt]$\frac{dv}{dt}=v\,\alpha(i,v)+\mathcal{N}(5,1)$};
\node[state] (neutral) at (0,-1.5) {\textbf{Neutral (gear=0)}\\[1pt]$\frac{dv}{dt}=-0.0005v^2+\mathcal{N}(0,0.5)$\\$\frac{dw}{dt}=-0.1$};

\draw[trans, bend left=15] (active.south east) .. controls +(0.5,-0.2) and +(0.5,0.2) .. node[right, text width=1.8cm, align=left] {$v\ge s_i $ gear=0, next=i+1, $w=0.8$} (neutral.east);
\draw[trans, bend left=15] (neutral.north west) .. controls +(-0.5,0.2) and +(-0.5,-0.2) .. node[left, text width=1.8cm, align=right] {$w<0$ gear=next} (active.west);
\end{tikzpicture}
}
\end{minipage}\hfill
\begin{minipage}{0.41\textwidth}
\caption*{\textbf{Gearbox} 
Dynamics are defined by $\alpha(i,v) = \frac{1}{1 + \frac{(v-p_i)^2}{25}}$ with $p_i = (5, 15, 25, 40, 60)$. 
In the reachability  optimization the loss function is set as the marginal $CDF(v_i \leq 16) \hspace{2mm} \forall i$ in order to maximize the probability of the car keeping always a velocity under $16$.}
\end{minipage}

\end{figure*}
\begin{figure*}

\begin{minipage}{0.45\textwidth}
\centering
\adjustbox{scale=0.90, valign=t}{
\begin{tikzpicture}[
    node distance=1cm,
    every node/.style={font=\scriptsize, align=center},
    state/.style={rectangle, rounded corners, draw=black, thick,
                  text width=3.5cm, minimum height=1.5cm, inner sep=1.5pt, fill=gray!10},
    trans/.style={-{Stealth[length=1.5mm,width=0.9mm]}, very thick}
]
\node[state] (fall) {\textbf{Fall (mode -1)}\\[1pt]$\frac{dv}{dt}=-g +\mathcal{N}(0,0.1)$};
\node[state, below=of fall] (refl) {\textbf{Reflection (mode 1)}\\[1pt]$\frac{dv}{dt}=- g - \frac{Rv + \frac{H}{C}}{m}+\mathcal{N}(0,0.1)$};

\draw[trans] (fall.south east) .. controls +(0.5,-0.2) and +(0.5,0.2) .. node[right, text width=1.5cm, align=left] {$H \leq 0$} (refl.east);
\draw[trans] (refl.north west) .. controls +(-0.5,0.2) and +(-0.5,-0.2) .. node[left, text width=1.5cm, align=right] {$H > 0$} (fall.west);
\end{tikzpicture}
}
\end{minipage}\hfill
\begin{minipage}{0.41\textwidth}
\caption*{\textbf{Controlled bouncing ball.} 
The position of a ball of mass $m= 7$ varies accordingly to $\frac{dH}{dt}= v +\mathcal{N}(0,\epsilon)$, starting from initial height $H_0\sim \mathcal{N}(9, 1)$, with  $g = 9.81$, $m = 7$. 
The reachability optimization target is expressed by $1 - CDF(H_i \geq 7)$ where $H_i$ is the maximum height achieved by the ball after bouncing.}
\end{minipage}

\label{fig:hybrid_systems}
\end{figure*}

\subsection{Analysis of reachability properties} \label{app:reach_prop}

The target of the optimization in the framework of the reachability properties is computed with respect to the DeGAS semantics. 
To obtain some qualitative insights on how this relate to the real system and the SOGA semantics, and on how the error may accumulate on long execution traces, we consider the example with largest iterations number, Controlled Bouncing Ball, and plot the mean trajectories computed by DeGAS, SOGA and randomly sampled trajectories from the real system. 
Figure \ref{plot:comparison_SOGA_real} shows that there is no significant discrepancy induced by the smoothing, even after a large number of iterations.
Moreover, despite the error induced by the Gaussian approximation, there is a substantial agreement among the mean trajectories estimated by SOGA and DeGAS and the simulated trajectories.

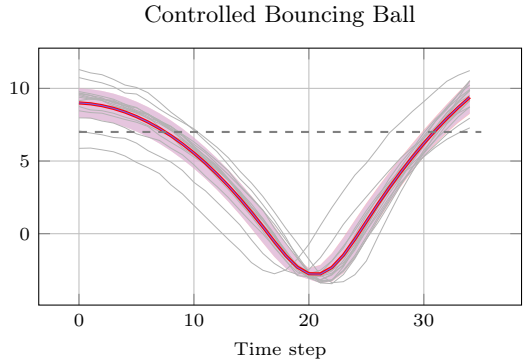
\begin{figure}[t]
    \centering
\begin{tikzpicture}
\begin{groupplot}[
    group style={group size=2 by 2, horizontal sep=0.9cm, vertical sep=1.7cm},
    width=8cm, height=5cm,
    xlabel={Time step},
    ylabel={},
    grid=both,
    ticklabel style={font=\scriptsize},
    label style={font=\scriptsize},
    title style={font=\small, yshift=-1pt},
    legend style={
        at={(0.5,-0.15)},
        anchor=north,
        legend columns=-1,
        font=\scriptsize,
    },
]

\newcommand{\plotcase}[4]{%
    \pgfplotstableread[col sep=comma]{#2}\rawtraj
    \pgfplotstableread[col sep=comma]{#3}\preddata
    \pgfplotstableread[col sep=comma]{#4}\initdata
    
    \pgfplotstablegetcolsof{\rawtraj}
    \pgfmathtruncatemacro{\numcols}{\pgfplotsretval-1}
    
    \foreach \r in {0,...,14} { 
        \addplot[black!30, very thin] table [
            x expr=\coordindex,
            y expr=\thisrowno{\r}
        ] {\rawtraj};
    }
    \addlegendimage{gray!60, thin}

    \addplot[very thick, red] table [x expr=\coordindex, y index=0] {\preddata};

    \addplot[very thin, blue!60] table [x expr=\coordindex, y index=0] {\initdata};

    \addplot[name path=upper, draw=none]
        table [x expr=\coordindex, y expr=\thisrowno{0}+\thisrowno{1}] {\preddata};
    \addplot[name path=lower, draw=none]
        table [x expr=\coordindex, y expr=\thisrowno{0}-\thisrowno{1}] {\preddata};
    \addplot[red!50, opacity=0.5] fill between[of=upper and lower];

    \addplot[name path=upper, draw=none]
        table [x expr=\coordindex, y expr=\thisrowno{0}+\thisrowno{1}] {\initdata};
    \addplot[name path=lower, draw=none]
        table [x expr=\coordindex, y expr=\thisrowno{0}-\thisrowno{1}] {\initdata};
    \addplot[blue!20, opacity=0.5] fill between[of=upper and lower];
}


\nextgroupplot[title={Controlled Bouncing Ball}]
\plotcase{green!60!black}{csv_files/bouncing_ball_opt_traj.csv}{csv_files/bouncing_ball2_opt.csv}{csv_files/bouncing_ball_SOGA.csv}
\addplot[dashed, gray, thick, domain = 0:35]{7};
\end{groupplot}

\end{tikzpicture}
\caption{DeGAS semantics (red) compared to SOGA semantics (blue) and real system (gray trajectories), all considering the optimized parameters obtained by DeGAS optimization of reachability properties}
\label{plot:comparison_SOGA_real}
\end{figure}

\subsection{Comparison with SOGA} \label{app:SOGAcomp}

For every program in Table \ref{tab:vi_mcmc_degas} we compared the mean value of a target variable computed in DeGAS and in SOGA for fixed value of the parameters, to asses the error introduced by the smoothing. Our results, reported in Table \ref{tab:means}, show that for the smoothing hyperparameters used in our experimental setting ($\smoothEps = 1e-3, \smoothDelta = \sqrt{\smoothEps}$) no error is introduced up to third decimal digit.

\begin{table}[H]
\centering
\begin{adjustbox}{width=1\textwidth}
\begin{tabular}{l|c|c|c|l}
\toprule
\textbf{Model} & \textbf{Target} & \textbf{DeGAS} & \textbf{SOGA} & \textbf{Parameters}\\
\midrule
Bernoulli & $y$ & $0.810$ & $0.810$ & $p = 0.794$ \\
Burglary  & $called$ & $0.237$ & $0.237$ & $pe = 0.063; pb = 0.132$\\
ClickGraph & $sim$ & $0.617$ & $0.617$ & $p = 0.646 $ \\
ClinicalTrial  & $ytreated$ & $0.789$ & $0.789$ & $pe = 0.054; pc = 0.797; pt= 0.778$ \\
Grass  & $wetGrass$ & $0.666$ & $0.666$ & $pcloudy=0.546; p1 = 0.705; p2 = 0.910; p3 = 0.864$ \\
MurderMistery  & $withGun$ & $0.716$ & $0.716$ & $palice = 0.138$ \\
SurveyUnbiased  & $ansb2$ & $0.713$ & $0.713$ & $bias1 = 0.778; bias2 =0.723$ \\
TrueSkills  & $perfC$ & $119.021$ & $119.021$ & $pa = 105.990; pb = 89.931; pc = 119.021$\\
TwoCoins  & $both$ & $0.158$ & $0.158$ & $first = 0.367; second = 0.367$\\
AlterMu & $y$ & $4.530$ & $4.530$ & $p1 = 0.617; p2 = 0.617; p3 = -3.388$ \\
AlterMu2  & $y$ & $3.553$ & $3.553$ & $muy = 3.553; vary = 2.838$\\
NormalMixtures & $y$ & $-3.192$ & $-3.192$ & $theta = 0.718; p1 = -4.788; p2 = 0.627$ \\
PID & $v$ & $-0.008$ & $-0.008$ & $s0 = 46.500; s1 = -23.496; s2 = -0.495$ \\

\bottomrule
\end{tabular}
\end{adjustbox}
\caption{Comparison among the mean values of a selected target variable computed by DeGAS and SOGA for the programs of Table \ref{tab:vi_mcmc_degas} and fixed values of the parameters.}
\label{tab:means}
\end{table}

\end{document}